\documentclass[10pt,conference]{IEEEtran}

\usepackage{cite}
\usepackage{amsmath,amssymb,amsfonts}
\usepackage{algorithmic}
\usepackage{graphicx}
\usepackage{textcomp}
\usepackage{xcolor}
\def\BibTeX{{\rm B\kern-.05em{\sc i\kern-.025em b}\kern-.08em
    T\kern-.1667em\lower.7ex\hbox{E}\kern-.125emX}}

\usepackage{standalone}

\usepackage{graphicx}
\usepackage{pdfpages}

\usepackage{colortbl}
\usepackage{xcolor}
\usepackage{array}
\usepackage{empheq}
\usepackage{changepage}
\usepackage{tikz}
\usepackage{listings}

\definecolor{light}{RGB}{222,235,247}
\definecolor{dark}{RGB}{158,202,225}
\definecolor{darker}{RGB}{49,130,189}
\definecolor{beaublue}{RGB}{189,212,230}
\definecolor{deepGreen}{RGB}{29, 149, 63}
\definecolor{lightgray}{RGB}{204, 204, 204}

\newcolumntype{a}{>{\columncolor{light}}r}
\newcolumntype{b}{>{\columncolor{dark}}r}
\usepackage[ruled,lined,linesnumbered,vlined]{algorithm2e}
\definecolor{commentgray}{gray}{0.5}

\SetCommentSty{mycommfont}

\usepackage{amsmath}
\usepackage{amsopn}
\usepackage{amsthm}
\usepackage{nccmath} 
\usepackage{url} 
\usepackage{xurl}
\usepackage{multicol}
\usepackage{multirow}
\usepackage{setspace} 
\usepackage{xspace}
\usepackage{listings}
\usepackage{stmaryrd} 
\usepackage{todonotes} 
\usepackage{booktabs}
\usepackage{threeparttable}
\usepackage[utf8x]{inputenc}
\usepackage[normalem]{ulem}	
\usepackage{balance}
\usepackage{subcaption}
\usepackage[T1]{fontenc}
\usepackage{tcolorbox}

\usepackage[%
  hidelinks
]{hyperref}

\usepackage[scaled=0.87]{beramono}

\usepackage[scaled]{helvet}
\usepackage{enumitem}
\usepackage{pdftexcmds}
\usepackage{makecell}
\usepackage[labelformat=simple]{subcaption}

\newcommand{\projtext}{CDD}

\newcommand{\proj}{\projtext\xspace}
\newcommand{\projfull}{\textsf{Counter-Based Delta Debugging}\xspace}
\newcommand{\codeurl}{\url{https://github.com/uw-pluverse/perses}}
\makeatletter

\newcommand{\Rmnum}[1]{\expandafter\@slowromancap\romannumeral #1@}
\makeatother

\newcommand{\gain}{Expected Reduction Gain\xspace}

\newcommand{\gcc}{GCC\xspace}
\newcommand{\llvm}{LLVM\xspace}

\newcommand{\java}{Java\xspace}

\newcommand{\creduce}{C-Reduce\xspace}
\newcommand{\jreduce}{J-Reduce\xspace}
\newcommand{\perses}{Perses\xspace}
\newcommand{\vulcan}{Vulcan\xspace}
\newcommand{\probdd}{ProbDD\xspace}
\newcommand{\probddnor}{ProbDD-no-random\xspace}
\newcommand{\myalgorithm}{\proj}
\newcommand{\hdd}{HDD\xspace}
\newcommand{\deltadebugging}{Delta Debugging\xspace}
\newcommand{\chisel}{Chisel\xspace}

\newcommand{\ddmin}{\mycode{ddmin}}
\newcommand{\picireny}{Picireny\xspace}
\newcommand{\picire}{Picire\xspace}

\newcommand{\xml}{XML\xspace}
\newcommand{\basex}{Basex\xspace}
\newcommand{\xpress}{Xpress\xspace}

\newcommand{\searchspace}{\ensuremath{\mathbb{L}}\xspace}

\newcommand{\boolspace}{\ensuremath{\mathbb{B}}\xspace}

\newcommand{\pincreaserate}{\ensuremath{\frac{1}{1 - e^{-1}}}\xspace}
\newcommand{\pincreaseratevalue}{1.582\xspace}

\newcommand{\cBenchmarkNameShort}{\ensuremath{\text{BM}_\text{C}}\xspace}
\newcommand{\cBenchmarkNumValue}{20}
\newcommand{\cBenchmarkNum}{\cBenchmarkNumValue\xspace}
\newcommand{\cBenchmarkSmallestSize}{4,397\xspace}
\newcommand{\cBenchmarkLargestSize}{212,259\xspace}

\newcommand{\debloatBenchmarkNameShort}{\ensuremath{\text{BM}_\text{DBT}}\xspace}
\newcommand{\debloatBenchmarkNumValue}{10}
\newcommand{\debloatBenchmarkNum}{\debloatBenchmarkNumValue\xspace}
\newcommand{\debloatBenchmarkSmallestSize}{34,801\xspace}
\newcommand{\debloatBenchmarkLargestSize}{163,296\xspace}

\newcommand{\xmlBenchmarkNameShort}{\ensuremath{\text{BM}_\text{XML}}\xspace}
\newcommand{\xmlBenchmarkNumValue}{46}
\newcommand{\xmlBenchmarkNum}{\xmlBenchmarkNumValue\xspace}
\newcommand{\xmlBenchmarkUniqueNumValue}{8}
\newcommand{\xmlBenchmarkUniqueNum}{\xmlBenchmarkUniqueNumValue\xspace}
\newcommand{\xmlBenchmarkSmallestSize}{19,290\xspace}
\newcommand{\xmlBenchmarkLargestSize}{20,750\xspace}

\newcommand{\benchmarkNumTotal}{%
  \pgfmathparse{\cBenchmarkNumValue + \debloatBenchmarkNumValue + \xmlBenchmarkNumValue}
    \pgfmathprintnumber{\pgfmathresult}\xspace%
}

\newcommand{\allBenchmarkSizeAvgProbdd}{235\xspace}

\newcommand{\allBenchmarkSizeProbddVsDdminPvalue}{0.32\xspace}
\newcommand{\allBenchmarkSizeProbddVsProjPvalue}{0.42\xspace}

\newcommand{\allBenchmarkTimeAvgDdminValue}{2999}
\newcommand{\allBenchmarkTimeAvgProbddValue}{2189}
\newcommand{\allBenchmarkTimeAvgProjValue}{2102}

\newcommand{\allBenchmarkTimeAvgProbdd}{2,189\xspace}

\newcommand{\allBenchmarkNoTimeoutTimeDecRateProbddVsDdmin}{%
  \pgfmathparse{(1 - \allBenchmarkTimeAvgProbddValue/\allBenchmarkTimeAvgDdminValue)*100}
    \pgfmathprintnumber{\pgfmathresult}\%\xspace%
}
\newcommand{\allBenchmarkNoTimeoutTimeDecRateProjVsDdmin}{%
  \pgfmathparse{(1 - \allBenchmarkTimeAvgProjValue/\allBenchmarkTimeAvgDdminValue)*100}
    \pgfmathprintnumber{\pgfmathresult}\%\xspace%
}

\newcommand{\allBenchmarkTimeProbddVsDdminPvalue}{3e-09\xspace}
\newcommand{\allBenchmarkTimeProbddVsProjPvalue}{0.29\xspace}

\newcommand{\allBenchmarkQueryAvgDdminValue}{2752}
\newcommand{\allBenchmarkQueryAvgProbddValue}{1309}
\newcommand{\allBenchmarkQueryAvgProjValue}{1320}

\newcommand{\allBenchmarkQueryAvgProbdd}{1,309\xspace}

\newcommand{\allBenchmarkNoTimeoutQueryDecRateProbddVsDdmin}{%
  \pgfmathparse{(1 - \allBenchmarkQueryAvgProbddValue/\allBenchmarkQueryAvgDdminValue)*100}
    \pgfmathprintnumber{\pgfmathresult}\%\xspace%
}
\newcommand{\allBenchmarkNoTimeoutQueryDecRateProjVsDdmin}{%
  \pgfmathparse{(1 - \allBenchmarkQueryAvgProjValue/\allBenchmarkQueryAvgDdminValue)*100}
    \pgfmathprintnumber{\pgfmathresult}\%\xspace%
}

\newcommand{\allBenchmarkQueryProbddVsDdminPvalue}{9e-14\xspace}
\newcommand{\allBenchmarkQueryProbddVsProjPvalue}{0.70\xspace}

\newcommand{\typeComplement}{\textit{Complement}\xspace}
\newcommand{\typeRepeated}{\textit{Revisit}\xspace}
\newcommand{\typeOther}{\textit{Other}\xspace}

\newcommand{\cTotalNumberDdmin}{901,128\xspace}

\newcommand{\cComplementNumberDdmin}{490,896\xspace}
\newcommand{\cComplementPercentageDdmin}{54.48\%\xspace}
\newcommand{\cComplementNumberSuccessDdmin}{9\xspace}
\newcommand{\cComplementSuccessRateDdmin}{<0.01\%\xspace}

\newcommand{\cRepeatedNumberDdmin}{170,884\xspace}
\newcommand{\cRepeatedPercentageDdmin}{18.96\%\xspace}
\newcommand{\cRepeatedNumberSuccessDdmin}{222\xspace}
\newcommand{\cRepeatedSuccessRateDdmin}{0.13\%\xspace}

\newcommand{\cOtherNumberSuccessDdmin}{11,927\xspace}
\newcommand{\cOtherSuccessRateDdmin}{4.98\%\xspace}

\newcommand{\cSuccessNumberDdmin}{12,158\xspace}
\newcommand{\cSuccessNumberProbdd}{11,696\xspace}

\newcommand{\cSuccessRateProbdd}{4.28\%\xspace}

\newcommand{\debloatRepeatedNumberSuccessDdmin}{1,048\xspace}
\newcommand{\debloatComplementNumberSuccessDdmin}{855\xspace}
\newcommand{\debloatOtherNumberSuccessDdmin}{8,437\xspace}
\newcommand{\debloatSuccessNumberDdmin}{10,340\xspace}
\newcommand{\debloatSuccessNumberProbdd}{13,878\xspace}

\newcommand{\xmlRepeatedNumberSuccessDdmin}{2\xspace}
\newcommand{\xmlComplementNumberSuccessDdmin}{830\xspace}
\newcommand{\xmlOtherNumberSuccessDdmin}{1,485\xspace}
\newcommand{\xmlSuccessNumberDdmin}{2,317\xspace}
\newcommand{\xmlSuccessNumberProbdd}{2,176\xspace}

\newcommand{\allBenchmarkSizeAvgProbddNoshuffle}{236\xspace}
\newcommand{\allBenchmarkSizeProbddVsProbddNoshufflePvalue}{0.87\xspace}

\newcommand{\allBenchmarkTimeAvgProbddNoshuffle}{2,069\xspace}
\newcommand{\allBenchmarkTimeProbddVsProbddNoshufflePvalue}{0.15\xspace}

\newcommand{\allBenchmarkQueryAvgProbddNoshuffle}{1,238\xspace}
\newcommand{\allBenchmarkQueryProbddVsProbddNoshufflePvalue}{0.10\xspace}

\newcommand{\cBenchmarkProbddCalledNumValue}{6871}
\newcommand{\cBenchmarkProbddCalledNum}{6,871\xspace}
\newcommand{\cBenchmarkProbddNotoneminimalNumValue}{76}
\newcommand{\cBenchmarkProbddNotoneminimalNum}{\cBenchmarkProbddNotoneminimalNumValue\xspace}
\newcommand{\cBenchmarkProbddNotoneminimalRate}{%
  \pgfmathparse{(\cBenchmarkProbddNotoneminimalNumValue/\cBenchmarkProbddCalledNumValue)*100}
    \pgfmathprintnumber{\pgfmathresult}\%\xspace%
}
\newcommand{\cBenchmarkProbddNotoneminimalAvgPotential}{1.49\xspace}

\newcommand{\aka}{\hbox{\emph{a.k.a.}}\xspace}

\newcommand{\etal}{\hbox{\emph{et al.}}\xspace}
\newcommand{\vs}{\hbox{\emph{vs.}}\xspace}
\newcommand{\eg}{\hbox{\emph{e.g.}}\xspace}
\newcommand{\ie}{\hbox{\emph{i.e.}}\xspace}

\newcommand{\wrt}{\hbox{\emph{w.r.t.}}\xspace}
\newcommand{\etc}{\hbox{\emph{etc.}}\xspace}

\newcommand{\mycode}[1]{\texttt{#1}\xspace}
\newcommand{\myelement}[1]{\ensuremath{{l}_{#1}}}
\newcommand{\mysubset}{\ensuremath{S}\xspace}
\newcommand{\myprobability}[2]{\ensuremath{p_{#1}^{#2}}}
\newcommand{\mysize}[2]{\ensuremath{s_{#1}^{#2}}}

\newcommand{\HighlightBlack}[1]{}

\newcommand{\myparagraph}[1]{
  \noindent \textit{\textbf{#1.}}\quad
}

\newcommand{\RoundNumber}{\ensuremath{r}\xspace}

\SetKwFunction{AlgReduce}{Reduce}

\SetKwFunction{AlgMain}{\proj}
\SetKwFunction{AlgPatch}{Patch}
\SetKwFunction{AlgTreeDiff}{TDiff}
\SetKwFunction{AlgListDiff}{LDiff}
\SetKwFunction{AlgPop}{Pop}
\SetKwFunction{AlgPush}{Push}
\SetKwFunction{AlgReduceTreeDiff}{ReduceDiffByTree}
\SetKwFunction{AlgReduceListDiff}{ReduceDiffByList}
\SetKwFunction{AlgReduceSeedWithDiff}{ReduceSeedWithDiff}
\SetKwFunction{AlgRemainingNodes}{GetRemainingNodes}
\SetKwFunction{AlgReducePlus}{ReducePlus}
\SetKwFunction{AlgReduceOptional}{ReduceOptional}
\SetKwFunction{AlgReduceRegular}{ReduceRegular}
\SetKwFunction{AlgGetAndRemoveLargestFrom}{GetAndRemoveLargestFrom}
\SetKwFunction{AlgChildren}{Children}
\SetKwFunction{AlgRemove}{CopyAndRemove}
\SetKwFunction{AlgReplace}{CopyAndReplace}
\SetKwFunction{AlgParent}{Parent}
\SetKwFunction{AlgDdmin}{\ddmin}
\SetKwFunction{AlgDdminWithOpts}{\ddmin-Opts}
\SetKwFunction{AlgDd}{dd}
\SetKwFunction{AlgParse}{ParseTree}
\SetKwFunction{AlgCopy}{Duplicate}
\SetKwFunction{AlgProbdd}{\probdd{}}
\SetKwFunction{AlgUpdateProbabilities}{UpdateProbabilities}
\SetKwFunction{AlgSelectSubset}{SelectSubset}
\SetKwFunction{AlgExpectedReductionGain}{ExpectedReductionGain}
\SetKwFunction{AlgRandomizedSort}{RandomizeThenSort}
\SetKwFunction{AlgCDD}{\myalgorithm}
\SetKwFunction{AlgReduceToSubset}{ReduceToSubset}
\SetKwFunction{AlgComputeSize}{ComputeSize}
\SetKwFunction{AlgPartition}{Partition}

\SetKwData{AlgMin}{Min}
\SetKwData{AlgMinIndex}{min\_index}
\SetKwData{AlgResult}{result}

\SetKwFunction{Encode}{Encode}
\SetKwFunction{CompactEncode}{CompactEncode}
\SetKwFunction{Decode}{CompactDecode}
\SetKwFunction{Refresh}{RefreshCache}

\SetKwFunction{BoundedBFS}{BoundedBFS}
\SetKwFunction{IsKleene}{IsKleene}
\SetKwFunction{RuleType}{Rule}
\SetKwFunction{ExpectedRuleType}{ExpectedRule}
\SetKwFunction{QuantifiedRuleType}{QuantifiedRule}

\SetKwData{Candidates}{candidates}
\SetKwData{Predicate}{pred}

\SetKwData{AlgTrueLiteral}{T}
\SetKwData{AlgFalseLiteral}{F}
\SetKwData{AlgMin}{min}
\SetKwData{AlgCache}{cache}
\SetKwData{AlgPrev}{prev}
\SetKwData{AlgWorklist}{worklist}
\SetKwData{AlgNext}{pending}
\SetKwData{AlgLargest}{largest}

\SetKw{Continue}{continue}
\SetKw{Break}{break}

\SetKwProg{Fn}{Function}{:}{}
\newtheorem{thm}{Theorem}[section]
\newtheorem{lem}[thm]{Lemma}

\newcommand\ignore[1]{}

\newcommand{\program}{\ensuremath{L}\xspace}
\newcommand{\listinput}{\program}
\newcommand{\testinput}{\ensuremath{I}\xspace}

\newcommand{\pinit}{\ensuremath{\myprobability{0}{}}\xspace}
\newcommand{\property}{\ensuremath{\psi}\xspace}

\newcommand{\variant}[1]{v$_{#1}$}

\usepackage{tikz}
\newcommand*\circled[1]{\tikz[baseline=(char.base)]{
		\node[shape=circle,draw,inner sep=1pt] (char) {#1};}}

\usetikzlibrary{calc}
\usepackage[framemethod=TikZ]{mdframed}
\usepackage{adjustbox}

\usepackage{forest}

\newcommand{\mt}[1]{%
	\tikz[overlay,remember picture,baseline] \coordinate (#1) at (0,0) {};}

\newcommand{\highlightGray}[2]{%
	\draw[lightgray,line width=8pt,opacity=0.7]%
	([yshift=2pt]#1) -- ([yshift=2pt]#2);%
}

\newcommand{\verticalline}{\unskip\ \vrule\ \ \ }

\newcounter{findingnum}
\newcounter{rqnum}

\newcommand{\finding}[1]{
  \begin{tcolorbox}[boxrule=0.5pt,arc=4pt,
        left=2pt,right=2pt,top=2pt,bottom=2pt,boxsep=2pt,
        colback=beaublue
        ]
    \textbf{Finding~\refstepcounter{findingnum}\thefindingnum}: #1
  \end{tcolorbox}
}

\def\firstchar#1#2\relax{#1}

\usepackage{cleveref} 

\Crefname{algocf}{Algorithm}{Algorithms}
\crefname{algocf}{Algorithm}{Algorithms}

\Crefname{algorithm}{Algorithm}{Algorithms}
\crefname{algorithm}{Algorithm}{Algorithms}

\crefname{appendix}{Appendix}{Appendices}
\Crefname{appendix}{Appendix}{Appendices}

\Crefname{figure}{Fig.}{Figs.}
\crefname{figure}{Fig.}{Figs.}

\crefname{listing}{Listing}{Listings}
\Crefname{listing}{Listing}{Listings}

\Crefname{table}{Table}{Tables}
\crefname{table}{Table}{Tables}

\crefname{thm}{Lemma}{Lemmas}
\Crefname{thm}{Lemma}{Lemmas}

\crefname{equation}{Equation}{Equations}
\Crefname{equation}{Equation}{Equations}

\crefformat{chapter}{\S~#2#1#3}
\crefmultiformat{chapter}{\S\S~#2#1#3}{ and~#2#1#3}{, #2#1#3}{, and~#2#1#3}

\crefformat{section}{\S~#2#1#3}
\crefmultiformat{section}{\S\S~#2#1#3}{ and~#2#1#3}{, #2#1#3}{, and~#2#1#3}

\begin{document}
\title{Toward a Better Understanding of Probabilistic Delta Debugging
}

\makeatletter
\newcommand{\linebreakand}{%
\end{@IEEEauthorhalign}
\hfill\mbox{}\par
\mbox{}\hfill\begin{@IEEEauthorhalign}
}
\makeatother

\author{
	\IEEEauthorblockN{
		Mengxiao Zhang\IEEEauthorrefmark{1},
		Zhenyang Xu\IEEEauthorrefmark{1},
		Yongqiang Tian\IEEEauthorrefmark{2},
		Xinru Cheng\IEEEauthorrefmark{1},
		and Chengnian Sun\IEEEauthorrefmark{1}
	}
	\IEEEauthorblockA{
		\IEEEauthorrefmark{1}School of Computer Science, University of
		Waterloo, Waterloo, Canada\\
		Emails: m492zhan@uwaterloo.ca, zhenyang.xu@uwaterloo.ca,
		x59cheng@uwaterloo.ca, cnsun@uwaterloo.ca
	}
	\IEEEauthorblockA{
		\IEEEauthorrefmark{2}Department of Computer Science and
		Engineering,\\
		The Hong Kong University of Science and Technology, Hong Kong,
		China\\
		Email: yqtian@ust.hk
	}
}

\maketitle


\begin{abstract}
	Given a list \program of elements and a property \property that \program exhibits,
\ddmin is a classic test input minimization algorithm
that aims to automatically remove \property-irrelevant elements from
\program.
This algorithm has been widely adopted in domains such as test input
minimization and
software debloating.
%
Recently, \probdd, a variant of \ddmin,
has been proposed and achieved state-of-the-art performance.
By
employing Bayesian optimization, \probdd estimates the probability of each
element in \program
being relevant to \property, and statistically
decides which and how many elements should be deleted together each time.
However, the theoretical probabilistic model of \probdd is rather intricate,
and the underlying details for the superior performance of \probdd
have not been adequately explored.

In this paper, we conduct the first in-depth theoretical analysis of \probdd,
clarifying the trends in probability and subset size changes and simplifying
the probability model.
We complement this analysis with empirical
experiments, including success rate analysis, ablation studies, and
examinations of trade-offs and limitations, to further comprehend and
demystify this state-of-the-art algorithm. Our success rate analysis reveals
how \probdd effectively addresses bottlenecks that slow down \ddmin by
skipping inefficient queries that attempt to delete complements of subsets
and previously tried subsets. The ablation study illustrates that
randomness in \probdd has no significant impact on efficiency. These
findings provide valuable insights for future research and applications of
test input minimization algorithms.

Based on the findings above, we propose \proj, a simplified version of
\probdd, reducing the complexity in both theory and implementation. \proj
assists in \circled{1} validating the correctness of our key findings, \eg, that
probabilities in \probdd essentially serve as monotonically increasing
counters for each element, and \circled{2} identifying the main factors that
truly contribute to \probdd's superior performance. Our comprehensive
evaluations across \benchmarkNumTotal benchmarks in test input
minimization and software debloating demonstrate that \proj can achieve
the same performance as \probdd, despite being much simplified.
\end{abstract}


\section{Introduction}
\label{sec:intro}
\deltadebugging~\cite{zeller2002simplifying} is a seminal family
of
algorithms designed for software debugging, among which \ddmin stands
out as a classic test input minimization (\aka, test input
reduction) algorithm.
Given a list \program of elements (modeling the test input)
and a property \property that \program exhibits,
\ddmin aims to remove elements in \program
that are irrelevant to \property,
such that
the resulting list
is smaller than \program yet still satisfies \property.
The
\ddmin algorithm
plays a crucial role in software testing, debugging and
maintenance~\cite{gccReductionGuide, llvmReductionGuide,
	webkit, asfbugzilla, mozillabugzilla},
since compact and informative bug-triggering inputs
are easier for developers to effectively identify root causes
than large bug-triggering inputs with bug-irrelevant
information~\cite{sigplan,ccmd,kitten,sun2016finding}.

To minimize a test input \testinput that satisfies \property,
\ddmin has been used in two primary manners.
In the first manner,
\testinput is initially segmented into a list, denoted as
\listinput, which could be segmented based on characters, tokens, lines, \etc Subsequently,
\ddmin is directly applied to
\listinput~\cite{zeller2002simplifying,donaldson2021test}.
Alternatively,
\ddmin serves as a pivotal component within advanced,
structure-aware test input minimization algorithms,
including \perses~\cite{sun2018perses},
\hdd~\cite{misherghi2006hdd},
\creduce~\cite{regehr2012test}, and
\chisel~\cite{heo2018effective}.
These algorithms leverage the inherent structures of
\testinput to expedite the minimization process or further reduce its size.
Generally, these algorithms initiate by parsing
\testinput into a tree structure, such as a parse tree.
They then iteratively extract a list \listinput of tree nodes
from the tree using heuristics and apply
\ddmin to \listinput to gradually condense the tree.
Both manners underscore the fundamental role of
\ddmin as the cornerstone of test input minimization.

In the past years,
different variants of
\ddmin have been proposed
to improve its
performance
~\cite{heo2018effective,wang2021probabilistic,gharachorlu2018avoiding,
 hodovan2016practical, zhou2024wdd},
among which Probabilistic Delta Debugging
(\probdd)~\cite{wang2021probabilistic} is the state of the art, with notable
superiority to other
algorithms~\cite{zeller2002simplifying, heo2018effective}.
When reducing \listinput,
\probdd utilizes a theoretical probabilistic model
based on
Bayesian
optimization to predict
 how likely every element in \program is essential to
 preserve
the property \property, by
assigning a probability to each element.
%
%
\probdd prioritizes
deleting
elements with lower
probabilities, as such elements generally have a lower possibility of being
\property-relevant. Before each deletion
attempt, an optimal subset of elements is
determined
by maximizing the \gain.\footnote{In each attempt, the
\gain
 is defined as
 the expected number
 of elements removed.
 Higher \gain is preferred, as it indicates an expectation to delete more
 elements through this attempt.}
If the deletion of this subset fails to preserve \property,
the probabilistic model increases the probability
assigned to each element in the subset.
As reported~\cite{probdd}, aided by such a probabilistic model,
\probdd  significantly outperforms \ddmin by
reducing the execution time and the query number.\footnote{A query is
a run of the property
test \property.}

However, this probabilistic model in
\probdd is rather intricate,
and the underlying mechanisms for its superior performance
have not been adequately studied.
The original paper of
\probdd merely
showed its performance numbers without deep
ablation analysis on such achievements.
Specifically, the following questions are important
to the research field of test input minimization,
but have not been answered yet.

\begin{enumerate}[topsep=0pt, leftmargin=*]
	\item What role do probabilities play in
	\probdd, and can they be simplified without impacting
	performance?
	\item What specific bottlenecks does \probdd overcome to achieve
	improvement compared to \ddmin?
    	\item How does randomness in \probdd contribute to
    	the performance improvement?
	\item What are the potential limitations of \probdd?
\end{enumerate}
Gaining a deeper understanding of the state of the art, \ie, \probdd, is
highly valuable for test input minimization tasks.
By clarifying
the intrinsic reasons behind its superiority, we can facilitate researchers to
understand
the essence of the probabilistic model, as well as its strengths and
limitations. Such demystification, in our view, paves the way for
enlightening future research and guides users to more
effectively apply
\ddmin and its variants for test input minimization.

To this end, we conduct the first in-depth analysis of \probdd, starting by
theoretically simplifying its probabilistic model. In the original \probdd,
probabilities are used to calculate the \gain,
which is subsequently used to determine the next subset size.
However, this process necessitates iterative calculations, impeding the
simplification and
comprehension of \probdd.
In our study,
we initially establish the analytical
correlation
between the probability and subset size, allowing for probabilities and
subset sizes
to be explicitly calculated through formulas, thus eliminating the need for
iterative updates.
Further, through mathematical derivation,
we
discover that the probability and subset size can be considered nearly
independent, each varying at an approximate ratio on their own. By
theoretical prediction, the probability increases
approximately by a factor of
\pincreaserate ($\approx$ \pincreaseratevalue), and the subset size
can be deduced by this probability, thus providing the potential for
simplifying
\probdd.

Building upon our theoretical analysis, we conducted extensive evaluations
of \ddmin, \probdd, and \proj across \benchmarkNumTotal diverse
benchmarks. The experimental results confirm
the correctness of our theoretical analysis, demonstrating how \probdd
addresses bottlenecks in \ddmin by skipping inefficient queries, reveals the
impact of randomness on results, and highlights the limitations of \probdd.
These findings provide valuable guidance for future research and the
development of test input minimization algorithms.

Based on the aforementioned analysis, we propose \projfull (\proj), a
simplified version of \probdd, to explain
\probdd's
high performance. By replacing probabilities with counters, \proj eliminates
the probability computations required by \probdd, thus reducing theoretical
and implementation complexity. Our experiments demonstrate that \proj
aligns with \probdd in both effectiveness and efficiency, which validates
our previous analysis and findings.

\myparagraph{Key Findings} Through both theoretical analysis
and empirical
experiments, our key findings are:

\begin{enumerate}[topsep=0pt, leftmargin=*]

	\item
	Through theoretical derivation, the probabilities in
	\probdd essentially serve as
	monotonically increasing counters, and can be simplified.
    This suggests that the probability mechanism itself may
    not be a critical factor in \probdd's superior performance.

	\item
	The performance bottlenecks addressed by \probdd are inefficient
	deletion attempts on complements of subsets and
	previously tried subsets, which should be considered to
	enhance efficiency.

	\item
	Randomness in
	\probdd has no significant
	impact on the performance.
	Test input minimization is an NP-complete problem,
	and randomness in \probdd does not produce smaller results.

	\item
	\probdd is faster than \ddmin, but at the cost of not guaranteeing
	1-minimality.\footnote{A list is considered to have 1-minimality if
	removing any single element from it results in the loss of
	its property.}
	The trade-off between effectiveness and
	efficiency is inevitable, and should be
	leveraged accordingly in different scenarios.


\end{enumerate}

\myparagraph{Contributions} We make the following major
contributions.

\begin{itemize}[leftmargin=*]
\item We perform the first in-depth theoretical analysis for \probdd, the
state-of-the-art algorithm in test input minimization tasks,
and identify
the latent correlation between the subset size and the
probability of elements.


\item
We propose \proj, a much
simplified version of \probdd.

\item
We evaluate \ddmin, \probdd and \proj on
\benchmarkNumTotal
benchmarks, validating the
correctness of our theoretical analysis.
Additional experiments and statistical analysis on \probdd
further explain its superior
performance, reveal the effectiveness of randomness,
and demonstrate the limitations of \probdd.

\item
To enable future research on test input minimization, we release the
artifact publicly for replication\cite{zhang2024cddartifact}. Additionally, we
have integrated \proj into the \perses project, available at
\codeurl.

\end{itemize}

\myparagraph{Paper Organization}
The remainder of the paper is structured as follows: \cref{sec:preliminaries}
introduces the symbols used in this study
and detailed workflow of \ddmin and  \probdd.
\cref{sec:finding_probability_and_size} presents our
in-depth analysis on
\probdd, simplifying the model of probability and subset size.
\cref{sec:experimental_results}
describes empirical experiments and results, from which additional
findings are derived.
\cref{sec:implication} introduces \proj, which simplifies \probdd based on
our earlier findings while maintaining equivalent performance.
\cref{sec:related_work} illustrates related work
and \cref{sec:conclusion}  concludes this study.


\section{Preliminaries}
\label{sec:preliminaries}
To facilitate comprehension,
\cref{tab:symbols} lists all the important symbols used in this paper.
Next, this section introduces \ddmin and \probdd,
with the running example shown in \cref{fig:running_example}.

\begin{table}[htbp]
	\footnotesize
\caption{The symbols used in this paper.
}
\label{tab:symbols}
\begin{tabular}{@{}ll|ll@{}}
	\toprule
	Symbol            & Description                      & Symbol                         &
	Description                                                                                 \\ \midrule
	\listinput        & the list to minimize             & \mysize{}{}                    & the
	size of \mysubset                                                                       \\
	\property         & the property to preserve         & $E(s)$                         &
	\makecell[tl]{\gain \\ with the first \mysize{}{}
	elements}                                 \\
	\myelement{i}     & the $i$-th element of \listinput & $e$
	& Euler's number                                                                              \\
	$\myelement{i}.p$ & the probability of \myelement{i} &
	\RoundNumber                   & the round
	number                                                                            \\
	$v_{i}$           & a variant of \listinput          &
	\mysize{\RoundNumber}{}        & the subset size in round
	\RoundNumber                                                       \\
	\mysubset         & a subset of \listinput           &
	\myprobability{\RoundNumber}{} & \multicolumn{1}{c}{\makecell[tl]{the
	probability of each \\ element in round \RoundNumber}} \\ \bottomrule
\end{tabular}
\end{table}
\lstset{
	language=Python,
	basicstyle=\ttfamily\tiny,
	keywordstyle=\color{black},
	stringstyle=\color{black},
	commentstyle=\color{black},
	breakatwhitespace=false,
	breaklines=true,
	captionpos=b,
	keepspaces=true,
	showspaces=false,
	showstringspaces=false,
	showtabs=false,
	tabsize=2,
    escapeinside={@@},
    mathescape,
}

\begin{figure}[h]
	\centering
\begin{subfigure}[b]{0.28\linewidth}
\begin{lstlisting}
@\myelement{1}@:import math, sys
@\myelement{2}@:input = sys.argv[1]
@\myelement{3}@:a = int(input)
@\myelement{4}@:b = math.e
@\myelement{5}@:c = 3
@\myelement{6}@:d = pow(b, a)
@\myelement{7}@:c = math.log(d, b)
@\myelement{8}@:crash(c)
\end{lstlisting}
	\caption{Original.}
	\label{subfig:running_example:original}
\end{subfigure}
\hfil
\verticalline
\begin{subfigure}[b]{0.28\linewidth}
\begin{lstlisting}
@\myelement{1}@:import math, sys
@\myelement{2}@:input = sys.argv[1]
@\myelement{3}@:a = int(input)
@\myelement{4}@:b = math.e
@\myelement{5}@:@\mt{1s}@c = 3@\mt{1e}@
@\myelement{6}@:d = pow(b, a)
@\myelement{7}@:c = math.log(d, b)
@\myelement{8}@:crash(c)
\end{lstlisting}
	\caption{By \ddmin.}
	\label{subfig:running_example:ddmin}
\end{subfigure}
\hfil
\verticalline
\begin{subfigure}[b]{0.28\linewidth}
\begin{lstlisting}
@\myelement{1}@:@\mt{2s}@import math, sys@\mt{2e}@
@\myelement{2}@:@\mt{3s}@input = sys.argv[1]@\mt{3e}@
@\myelement{3}@:@\mt{4s}@a = int(input)@\mt{4e}@
@\myelement{4}@:@\mt{5s}@b = math.e@\mt{5e}@
@\myelement{5}@:c = 3
@\myelement{6}@:@\mt{6s}@d = pow(b, a)@\mt{6e}@
@\myelement{7}@:@\mt{7s}@c = math.log(d, b)@\mt{7e}@
@\myelement{8}@:crash(c)
\end{lstlisting}
	\caption{By \probdd.}
	\label{subfig:running_example:probdd}
\end{subfigure}
	\begin{tikzpicture}[remember picture, overlay]
		\highlightGray{1s}{1e}
		\highlightGray{2s}{2e}
		\highlightGray{3s}{3e}
		\highlightGray{4s}{4e}
		\highlightGray{5s}{5e}
		\highlightGray{6s}{6e}
		\highlightGray{7s}{7e}
	\end{tikzpicture}
	\caption{A running example in Python.
        \cref{subfig:running_example:original} shows the original
        program, represented as a list of 8 elements (\myelement{1},
        \myelement{2}, $\cdots$, \myelement{8}), in which \myelement{8} (\ie, \mycode{crash(c)})
        triggers the crash.
        \cref{subfig:running_example:ddmin} and \cref{subfig:running_example:probdd} show the minimized
        results by \ddmin and \probdd, with removed elements masked in
        \colorbox{lightgray}{gray}. Both minimized programs still trigger the
        crash. Note that \probdd cannot consistently guarantee the result in
        \cref{subfig:running_example:probdd} and might produce larger results,
        due to its inherent randomness.
    }
	\label{fig:running_example}
\end{figure}

\begin{table*}[ht]
    \centering
        \caption{Step-by-step outcomes from \ddmin on the running
        example.
        In each column, a variant is generated and tested against
        the property \property. These variants are sequentially generated from
        left to right.  The first row displays
        the variant identifier, and the second row displays
        round number \RoundNumber and subset size \mysize{}{}. In the
        following rows,
        the symbol ``\checkmark'' denotes an element is included by a
        certain variant,
        while \colorbox{lightgray}{gray} cells signify
        that the element have been removed.
        For the last row, \protect\AlgTrueLiteral indicates that the variant still
        preserves the property \property, whereas
        \protect\AlgFalseLiteral indicates not.
    }
        \label{table:running_example_results_ddmin}
            \resizebox{\textwidth}{!}{%
\begin{tabular}{cc|cc|cccccccc|cccccccccccccccccccc}
	\toprule
	\multicolumn{1}{c|}{Initial} & \multicolumn{1}{c|}{Variants}  &
	\variant{1}          & \variant{2}           & \multicolumn{1}{c}{\variant{3}}      &
	\multicolumn{1}{c}{\variant{4}}      & \multicolumn{1}{c}{\variant{5}}      &
	\multicolumn{1}{c}{\variant{6}}      & \variant{7}          & \variant{8}          &
	\variant{9}          & \variant{10}          & \multicolumn{1}{c}{\variant{11}}     &
	\multicolumn{1}{c}{\variant{12}}     & \multicolumn{1}{c}{\variant{13}}     &
	\multicolumn{1}{c}{\variant{14}}     & \multicolumn{1}{c}{\variant{15}}     &
	\multicolumn{1}{c}{\variant{16}}     & \multicolumn{1}{c}{\variant{17}}     &
	\multicolumn{1}{c}{\variant{18}}     & \variant{19}         & \variant{20}
	& \variant{21}         & \variant{22}         & \variant{23}                                 &
	\variant{24}                                 & \variant{25}                                 &
	\variant{26}                                 & \variant{27}                                 &
	\variant{28}                                 & \variant{29}                                 &
	\variant{30}                                 \\ \cmidrule(l){2-32}
	\multicolumn{1}{c|}{Element} & \multicolumn{1}{c|}{Round} &
	\multicolumn{2}{c|}{$r=1$ (\mysize{}{}=4)}                     &
	\multicolumn{8}{c|}{$r=2$ (\mysize{}{}=2)}

	                  &

	\multicolumn{20}{c}{$r=3$ (\mysize{}{}=1)}

	                                          \\
	 \midrule
	\multicolumn{1}{c|}{\myelement{1}}      &                                & \checkmark           &
	\multicolumn{1}{c|}{} & \multicolumn{1}{c}{\checkmark}
	&                                      &                                      &                                      &
	\multicolumn{1}{c}{} & \checkmark           & \checkmark           &
	\checkmark            & \multicolumn{1}{c}{\checkmark}
	&                                      &                                      &
	&                                      &                                      &
	&                                      & \multicolumn{1}{c}{} & \checkmark           &
	\checkmark           & \checkmark           & \checkmark
	& \multicolumn{1}{c}{}                         & \checkmark                                   &
	\checkmark                                   & \checkmark                                   &
	\checkmark                                   & \checkmark                                   &
	\checkmark                                   \\
	\multicolumn{1}{c|}{\myelement{2}}      &                                & \checkmark           &
	\multicolumn{1}{c|}{} & \multicolumn{1}{c}{\checkmark}
	&                                      &                                      &                                      &
	\multicolumn{1}{c}{} & \checkmark           & \checkmark           &
	\checkmark            &                                      &
	\multicolumn{1}{c}{\checkmark}       &
	&                                      &                                      &
	&                                      &                                      & \checkmark           &
	\multicolumn{1}{c}{} & \checkmark           & \checkmark           &
	\checkmark                                   & \checkmark                                   &
	\multicolumn{1}{c}{}                         & \checkmark                                   &
	\checkmark                                   & \checkmark                                   &
	\checkmark                                   & \checkmark                                   \\
	\multicolumn{1}{c|}{\myelement{3}}      &                                & \checkmark           &
	\multicolumn{1}{c|}{} &                                      &
	\multicolumn{1}{c}{\checkmark}       &
	&                                      & \checkmark           & \multicolumn{1}{c}{} &
	\checkmark           & \checkmark            &
	&                                      & \multicolumn{1}{c}{\checkmark}
	&                                      &                                      &
	&                                      &                                      & \checkmark           &
	\checkmark           & \multicolumn{1}{c}{} & \checkmark           &
	\checkmark                                   & \checkmark                                   &
	\checkmark                                   & \multicolumn{1}{c}{}                         &
	\checkmark                                   & \checkmark                                   &
	\checkmark                                   & \checkmark                                   \\
	\multicolumn{1}{c|}{\myelement{4}}      &                                & \checkmark           &
	\multicolumn{1}{c|}{} &                                      &
	\multicolumn{1}{c}{\checkmark}       &
	&                                      & \checkmark           & \multicolumn{1}{c}{} &
	\checkmark           & \checkmark            &
	&                                      &                                      &
	\multicolumn{1}{c}{\checkmark}       &
	&                                      &                                      &                                      &
	\checkmark           & \checkmark           & \checkmark           &
	\multicolumn{1}{c}{} & \checkmark                                   &
	\checkmark                                   & \checkmark                                   &
	\checkmark                                   & \multicolumn{1}{c}{}                         &
	\checkmark                                   & \checkmark                                   &
	\checkmark                                   \\
	\multicolumn{1}{c|}{\myelement{5}}      &                                & \multicolumn{1}{c}{} &
	\checkmark            &                                      &                                      &
	\multicolumn{1}{c}{\checkmark}       &                                      &
	\checkmark           & \checkmark           & \multicolumn{1}{c}{} &
	\checkmark            &                                      &
	&                                      &                                      &
	\multicolumn{1}{c}{\checkmark}       &
	&                                      &                                      & \checkmark           &
	\checkmark           & \checkmark           & \checkmark           &
	\multicolumn{1}{c}{\cellcolor[HTML]{CCCCCC}} &
	\multicolumn{1}{c}{\cellcolor[HTML]{CCCCCC}} &
	\multicolumn{1}{c}{\cellcolor[HTML]{CCCCCC}} &
	\multicolumn{1}{c}{\cellcolor[HTML]{CCCCCC}} &
	\multicolumn{1}{c}{\cellcolor[HTML]{CCCCCC}} &
	\multicolumn{1}{c}{\cellcolor[HTML]{CCCCCC}} &
	\multicolumn{1}{c}{\cellcolor[HTML]{CCCCCC}} &
	\multicolumn{1}{c}{\cellcolor[HTML]{CCCCCC}} \\
	\multicolumn{1}{c|}{\myelement{6}}      &                                & \multicolumn{1}{c}{} &
	\checkmark            &                                      &                                      &
	\multicolumn{1}{c}{\checkmark}       &                                      &
	\checkmark           & \checkmark           & \multicolumn{1}{c}{} &
	\checkmark            &                                      &
	&                                      &                                      &                                      &
	\multicolumn{1}{c}{\checkmark}       &
	&                                      & \checkmark           & \checkmark           &
	\checkmark           & \checkmark           & \checkmark
	& \checkmark                                   & \checkmark                                   &
	\checkmark                                   & \checkmark                                   &
	\multicolumn{1}{c}{}                         & \checkmark                                   &
	\checkmark                                   \\
	\multicolumn{1}{c|}{\myelement{7}}      &                                & \multicolumn{1}{c}{} &
	\checkmark            &                                      &
	&                                      & \multicolumn{1}{c}{\checkmark}       &
	\checkmark           & \checkmark           & \checkmark           &
	\multicolumn{1}{c|}{} &                                      &
	&                                      &                                      &
	&                                      & \multicolumn{1}{c}{\checkmark}
	&                                      & \checkmark           & \checkmark           &
	\checkmark           & \checkmark           & \checkmark
	& \checkmark                                   & \checkmark                                   &
	\checkmark                                   & \checkmark                                   &
	\checkmark                                   & \multicolumn{1}{c}{}                         &
	\checkmark                                   \\
	\multicolumn{1}{c|}{\myelement{8}}      &                                & \multicolumn{1}{c}{} &
	\checkmark            &                                      &
	&                                      & \multicolumn{1}{c}{\checkmark}       &
	\checkmark           & \checkmark           & \checkmark           &
	\multicolumn{1}{c|}{} &                                      &
	&                                      &                                      &
	&                                      &                                      &
	\multicolumn{1}{c}{\checkmark}       & \checkmark           &
	\checkmark           & \checkmark           & \checkmark           &
	\checkmark                                   & \checkmark                                   &
	\checkmark                                   & \checkmark                                   &
	\checkmark                                   & \checkmark                                   &
	\checkmark                                   & \multicolumn{1}{c}{}                         \\
	\midrule
	\multicolumn{2}{c|}{\property}                                & \AlgFalseLiteral     &
	\AlgFalseLiteral      & \multicolumn{1}{c}{\AlgFalseLiteral} &
	\multicolumn{1}{c}{\AlgFalseLiteral} & \multicolumn{1}{c}{\AlgFalseLiteral}
	& \multicolumn{1}{c}{\AlgFalseLiteral} & \AlgFalseLiteral     &
	\AlgFalseLiteral     & \AlgFalseLiteral     & \AlgFalseLiteral      &
	\multicolumn{1}{c}{\AlgFalseLiteral} & \multicolumn{1}{c}{\AlgFalseLiteral}
	& \multicolumn{1}{c}{\AlgFalseLiteral} &
	\multicolumn{1}{c}{\AlgFalseLiteral} & \multicolumn{1}{c}{\AlgFalseLiteral}
	& \multicolumn{1}{c}{\AlgFalseLiteral} &
	\multicolumn{1}{c}{\AlgFalseLiteral} & \multicolumn{1}{c}{\AlgFalseLiteral}
	& \AlgFalseLiteral     & \AlgFalseLiteral     & \AlgFalseLiteral     &
	\AlgFalseLiteral     & \AlgTrueLiteral                              &
	\AlgFalseLiteral                             & \AlgFalseLiteral                             &
	\AlgFalseLiteral                             & \AlgFalseLiteral                             &
	\AlgFalseLiteral                             & \AlgFalseLiteral                             &
	\AlgFalseLiteral                             \\ \bottomrule
\end{tabular}
        }
\end{table*}

\subsection{The \ddmin Algorithm}
The \ddmin algorithm~\cite{zeller2002simplifying} is the first algorithm to
systematically minimize a bug-triggering input to its essence, which has
been widely adopted in program reduction\cite{sun2018perses,
misherghi2006hdd, regehr2012test},
software debloating\cite{christi2017resource, heo2018effective} and
 test suites reduction\cite{groce2014cause, groce2016cause}.
It takes the following two inputs:
\begin{itemize}[leftmargin=*]
\item \program: a list of elements representing a bug-triggering input.
For example, \program can be a list of bytes, characters, lines, tokens, or parse
tree nodes extracted from the bug-triggering input.

\item \property: a property that \program has. Formally, \property can be defined as
a predicate that returns \AlgTrueLiteral if a list of elements preserves the property,
\AlgFalseLiteral otherwise.
\end{itemize}
%
and returns a minimal subset of \program that still preserves \property,
from which excluding any single element
will make the minimal subset lose \property.
This algorithm has been widely used in practice
to facilitate developers in debugging
\cite{regehr2012test,sun2018perses,donaldson2021test,picire}.
It generally consists of the following three steps.

\myparagraph{Initialize} Start by setting the initial subset size \mysize{}{} to
half of the
input list $L$, \ie, \mysize{}{}=$|L|/2$.

\myparagraph{Step 1: Minimize to Subset}
Partition $L$ into subsets of size $s$.
For each subset \mysubset,
check whether \mysubset alone satisfies \property. If yes, keep only
\mysubset and restart from Step 1 with $L = \mysubset$ and the subset
size as half of the new $L$;
otherwise, go to Step 2.

\myparagraph{Step 2: Minimize to Complement}
Partition $L$ into subsets of size $s$.
For each subset
\mysubset, check whether the
complement of \mysubset (\ie, $L/\mysubset = \{ e | e \in L
\wedge e \not\in \mysubset \}$) satisfies \property. If yes, keep the
complement of \mysubset and restart from Step 2 with $L = L /
\mysubset$;
otherwise,
go to Step 3.

\myparagraph{Step 3: Subdivide} If any of the remaining subsets has at
least two elements and thus can be further divided, halve the subset size, \ie, $\mysize{}{} = \mysize{}{} / 2$
and go back to Step 1. If no subset can be further divided (i.e., the subset
size is 1), \ddmin terminates and returns the remaining elements as the
result.

\myparagraph{Round Number \RoundNumber}
Note that we introduce a round number \RoundNumber at the second
column of
\cref{table:running_example_results_ddmin}.
Within each round, the list \listinput is divided into subsets of a fixed size,
on which Step 1 and Step 2 are applied.
A new round begins when no further progress can be made with the
current subset size.
This round number is \emph{not explicitly} present
in the original \ddmin algorithm but exists \emph{implicitly}. In subsequent
sections, we will also use this concept to introduce and simplify the
\probdd
algorithm.


\cref{table:running_example_results_ddmin}
illustrates the step-by-step minimization process of \ddmin with the running
example in \cref{fig:running_example}.
Initially, the input \listinput is [\myelement{1}, \myelement{2}, $\cdots$,
\myelement{8}].
The \ddmin algorithm iteratively
generates variants by gradually decreasing the subset size from 4 to 1.
\begin{enumerate}[leftmargin=*]
	\item Round 1 (\mysize{}{}=4). At the beginning, \ddmin
	splits $L$ into two subsets and generates two variants
	\variant{1} and \variant{2}.
	However, neither of them preserves \property.

	\item Round 2 (\mysize{}{}=2). Next, \ddmin continues to subdivide these
	two subsets into
	smaller ones, and generates eight variants (\ie, \variant{3},
	\variant{4}, $\cdots$, \variant{10})
	by using these subsets and their complements.
	Specifically, the first four variants (\variant{3}, \variant{4}, \variant{5},
	\variant{6}) are the subsets, and the next four variants
	(\variant{7}, \variant{8}, \variant{9}, \variant{10}) are the
	complements of these subsets.
	Again, none of these eight variants preserves \property.

	\item Round 3 (\mysize{}{}=1). Finally,
	\ddmin decreases subset size \mysize{}{} from 2 to 1, and generates more
	variants.
	This time, \variant{23}, which is the
	complement of the subset \{\myelement{5}\},
	preserves \property.
	Hence, the subset \{\myelement{5}\} is permanently removed from
	\listinput.
	Then for each of the remaining subsets \{\myelement{1}\},
	\{\myelement{2}\},
	$\cdots$, \{\myelement{8}\},
	\ddmin restarts testing the complement of each subset,
	\ie, from \variant{24} to \variant{30}.
	However, none of these variants preserves \property,
	and no subset can be further divided,
	so \ddmin terminates with the variant \variant{23} as the final result.
\end{enumerate}

\subsection{Probabilistic Delta Debugging (\probdd)}
Wang \etal~\cite{wang2021probabilistic} proposed the state-of-the-art
algorithm \probdd, significantly surpassing \ddmin in minimizing
bug-triggering programs on C compilers and
benchmarks in software debloating.
\probdd employs Bayesian optimization~\cite{pelikan1999boa} to model
the minimization problem. \probdd assigns a probability to each element in \listinput,
representing its likelihood of being essential for preserving the property
\property.
At each step during the minimization process, \probdd
selects a subset
of elements expected to yield the
highest \gain, and targets these elements in the subset for deletion.
In this section, we outline \probdd's workflow in \cref{alg:probdd}, paving
the way for a deeper
understanding and analysis of \probdd.

\newcommand{\currentMaxGain}{\texttt{currentMaxGain}\xspace}
\begin{algorithm}[h]
\footnotesize
\DontPrintSemicolon
\SetKwInput{KwData}{Input}
\SetKwInput{KwResult}{Output}
\caption{\protect\AlgProbdd{$\listinput, \psi$}
}
\label{alg:probdd}

\SetKwRepeat{Do}{do}{while}
\KwData{\listinput: a list to be minimized.
}
\KwData{$\psi: \searchspace \rightarrow \boolspace$: the property to be
preserved by \listinput.}
\KwData{$\pinit$: the initial probability given by the user.}
\KwResult{the minimized list that still exhibits the property $ \psi $.}


\tcp{Initialize the probability of each element with \pinit}
\lForEach{$\myelement{} \in \listinput$}{
    $\myelement{}.p \gets \pinit$ \label{alg:probdd:init_p}
}

\tcc{The round
	number \RoundNumber, initially 0. \RoundNumber is not explicitly used in
	the
	original \probdd algorithm. It is
displayed for demonstrating \probdd's implicit principles.}
$\RoundNumber \gets 0$ \label{line:probdd:round:init} \;
\label{alg:counterdd:begin}
\While{$\exists \myelement{} \in \listinput : \myelement{}.p < 1$}{
	\label{alg:probdd:check_termination}
	\tcp{Select elements from \listinput for deletion attempt.}
	$\mysubset \gets \AlgSelectSubset(\listinput)$ \;
	\label{alg:probdd:call_select_subset}
	\tcp{Check if removing the subset preserves the property}
	$\texttt{temp} \gets \listinput \setminus \mysubset$ \;
	\label{alg:probdd:remove_s_begin}
	\lIf{$\psi(\texttt{temp})$ = \AlgTrueLiteral}{
		$\listinput \gets \texttt{temp}$ \label{alg:probdd:remove_s}
	} \Else {
   		\tcp{Calculate the factor to update probabilities}
   		$\texttt{factor} \gets \frac{1}{1 - \prod_{\myelement{} \in \mysubset} (1
   		- \myelement{}.p)}$ \\
   		\label{alg:probdd:calculate_factor}
   		\tcp{Update the probabilities of elements in the subset}
   		\lForEach{$\myelement{} \in \mysubset$}{
   			$\myelement{}.p \gets \texttt{factor} \times \myelement{}.p$
   			\label{alg:probdd:update_probability}
   		}
	}
	\label{alg:probdd:remove_s_end}
    \If{\textit{All elements' probability have been
    updated}}{\label{line:probdd:round:test}
        \tcp{Move to the next round.}
        $\RoundNumber = \RoundNumber + 1$\label{line:probdd:round:incr} \;
    }
}
\Return \listinput \; \label{alg:counterdd:end}

\BlankLine

\Fn{\AlgSelectSubset{$\texttt{L}$}}{
\KwData{\listinput: a list of elements to be reduced.}
\KwResult{The subset of elements that maximizes the \gain.}
\tcc{Sort \listinput by ascending probability, with elements having the same
probability in random order.}
$\texttt{sortedL} \gets \AlgRandomizedSort(L)$ \\
\label{alg:probdd:select_subset_begin}
\label{alg:probdd:sort}
$\mysubset \gets \emptyset$ \\
$\currentMaxGain \gets 0$ \\
\ForEach{$ \myelement{} \in \texttt{sortedL}$}{
	$\texttt{tempSubset} \gets \mysubset \cup \{\myelement{}\}$ \\
	$\texttt{gain} \gets |\texttt{tempSubset}| \times \prod_{\myelement{} \in
		\texttt{tempSubset}} (1 - \myelement{}.p)$ \\
			\label{alg:probdd:gain}
	\If{$\texttt{gain} > \currentMaxGain$}{
		$\currentMaxGain \gets \texttt{gain}$ \\
		$\mysubset \gets \texttt{tempSubset}$
	} \lElse {
		\Break
	}
}
\Return \mysubset
\label{alg:probdd:select_subset_end}
}

\end{algorithm}

\myparagraph{Initialize (\cref{alg:probdd:init_p})}
In \listinput, \probdd assigns each element an initial
probability \pinit on \cref{alg:probdd:init_p},
representing the prior likelihood
that each element cannot be removed.

\myparagraph{Step 1: Select elements
(\cref{alg:probdd:call_select_subset},
\cref{alg:probdd:select_subset_begin}--\ref{alg:probdd:select_subset_end})}
First, \probdd sorts the
elements in \listinput by probability in ascending order on \cref{alg:probdd:sort},
and the order of elements with the same probability is determined
randomly.
Then, on \cref{alg:probdd:gain}, it calculates the subset to be removed in
the next attempt via the
proposed \gain $E(\mysize{}{})$, as shown in
 \cref{equation:expected_size_decrease},
with $E(\mysize{}{})$
denoting the expected gain obtained via removing the \emph{first
\mysize{}{} elements}
in $L$ selected for deletion, and
\myelement{i}.\myprobability{}{} denoting
the current probability of the $i$-th element in \listinput.
 \begin{align}
 	E(\mysize{}{}) = \mysize{}{} \times
 	\prod^{\mysize{}{}}_{i=1}{(1 -
 	\myelement{i}.\myprobability{}{})}
 	\label{equation:expected_size_decrease}
 \end{align}


 \noindent
 Note that \probdd has an invariant that the subset \mysubset chosen for deletion attempt is always
 the first \mysize{}{} elements in \listinput.
Every time, the first \mysize{}{*} elements are selected as the optimal
subset
\mysubset, where
\mysize{}{*} maximizes the \gain $E(\mysize{}{})$, elaborated as
 \cref{equation:maximize_expected_size_decrease}.
 \begin{align}
	\mysize{}{*} = \operatorname*{arg\,max}_{\mysize{}{}}
    E(\mysize{}{})
	\label{equation:maximize_expected_size_decrease}
\end{align}

\myparagraph{Step 2: Delete the Subset
(\cref{alg:probdd:remove_s_begin}-\ref{alg:probdd:remove_s_end})} If
\property is still
preserved after
the removal of \mysubset, \probdd removes subset \mysubset
on
\cref{alg:probdd:remove_s}, \ie, keeps only the
complement of \mysubset, and proceeds to Step 1.
%
If \property cannot be preserved after
the removal, on
\cref{alg:probdd:calculate_factor,alg:probdd:update_probability},
\probdd updates the probability of each element in the subset
\mysubset via \cref{equation:probability_updates}, and resumes at Step 1.
It is important to note that if an element \myelement{i} has been individually
deleted but failed, its probability \myelement{i}.\myprobability{}{} will be set
to 1, indicating that this element cannot be removed and will no longer be
considered for deletion.
 \begin{align}
	\myelement{i}.\myprobability{}{} \gets
	\frac{\myelement{i}.\myprobability{}{}}{1-\prod_{\myelement{}
	 \in
	\mysubset}
	{(1 - \myelement{}.\myprobability{}{})}}
	\label{equation:probability_updates}
\end{align}

\myparagraph{Step 3: Check Termination
(\cref{alg:probdd:check_termination})} If every element either has been
deleted, or possesses a probability of 1,
\probdd terminates. If not, it returns to Step 1.

\myparagraph{Round Number \RoundNumber}
Similar to the concept of rounds in \ddmin (see
\cref{table:running_example_results_ddmin}), \probdd also has an \emph{implicit}
round number
\RoundNumber, as introduced
on \cref{line:probdd:round:init} in \cref{alg:probdd} and
the second row of
\cref{table:running_example_results_probdd}.
During a round, the subset size is the same and every subset in \listinput is
attempted for deletion.
Once the probabilities of all elements have been updated, the next round
begins (\ie, $\RoundNumber \gets \RoundNumber + 1$ on
\cref{line:probdd:round:incr}).

\begin{table}[ht]
        \centering
        \caption{Step-by-step outcomes from \probdd on the running
        	example. Similar to
        	\cref{table:running_example_results_ddmin},
        	 round number, subset size and the details of
        	 each
        	 variants
        	 are presented. For each variant, the probability of
        	each element is noted
        	alongside.
        }
        \label{table:running_example_results_probdd}
            \resizebox{\linewidth}{!}{
\begin{tabular}{ccc|cccc|cccc|cccccccc}
	\toprule
	\multicolumn{2}{c|}{Initial}                             &
	\multicolumn{1}{c|}{Variants}  & \multicolumn{2}{c}{\variant{1}}      &
	\multicolumn{2}{c|}{\variant{2}}                                         &
	\multicolumn{2}{c}{\variant{3}}                           &
	\multicolumn{2}{c|}{\variant{4}}                           &
	\multicolumn{2}{c}{\variant{5}}                                         &
	\multicolumn{2}{c}{\variant{6}}                           &
	\multicolumn{2}{c}{\variant{7}}                           &
	\multicolumn{2}{c}{\variant{8}}                           \\ \midrule
	\multicolumn{1}{c|}{Element} & \multicolumn{1}{c|}{Prob} &
	\multicolumn{1}{c|}{Round} &
	\multicolumn{4}{c|}{$r=1$ (\mysize{}{}=4)}
	       &

	\multicolumn{4}{c|}{$r=2$ (\mysize{}{}=2)}
	              &

	\multicolumn{8}{c}{$r=3$ (\mysize{}{}=1)}

	                     \\
	 \midrule
	\multicolumn{1}{c|}{\myelement{1}}      & \multicolumn{1}{c|}{0.25}
	&                                & \multicolumn{1}{l}{}      & 0.37     &
	\checkmark                                   & \multicolumn{1}{c|}{0.37}
	&                                & \multicolumn{1}{c}{0.61} &
	\multicolumn{1}{c}{\checkmark} & \multicolumn{1}{c|}{0.61} &
	\checkmark                                   & \multicolumn{1}{c}{0.61} &
	\multicolumn{1}{c}{\checkmark} & \multicolumn{1}{c}{0.61} &
	\cellcolor[HTML]{CCCCCC}       & \cellcolor[HTML]{CCCCCC} &
	\cellcolor[HTML]{CCCCCC}       & \cellcolor[HTML]{CCCCCC} \\
	\multicolumn{1}{c|}{\myelement{2}}      &
	\multicolumn{1}{c|}{0.25}
	&                                & \checkmark                & 0.25     &
	\multicolumn{1}{l}{\cellcolor[HTML]{CCCCCC}} &
	\cellcolor[HTML]{CCCCCC}  & \cellcolor[HTML]{CCCCCC}       &
	\cellcolor[HTML]{CCCCCC} & \cellcolor[HTML]{CCCCCC}       &
	\cellcolor[HTML]{CCCCCC}  &
	\multicolumn{1}{l}{\cellcolor[HTML]{CCCCCC}} &
	\cellcolor[HTML]{CCCCCC} & \cellcolor[HTML]{CCCCCC}       &
	\cellcolor[HTML]{CCCCCC} & \cellcolor[HTML]{CCCCCC}       &
	\cellcolor[HTML]{CCCCCC} & \cellcolor[HTML]{CCCCCC}       &
	\cellcolor[HTML]{CCCCCC} \\
	\multicolumn{1}{c|}{\myelement{3}}      & \multicolumn{1}{c|}{0.25}
	&                                & \checkmark                & 0.25     &
	\multicolumn{1}{l}{\cellcolor[HTML]{CCCCCC}} &
	\cellcolor[HTML]{CCCCCC}  & \cellcolor[HTML]{CCCCCC}       &
	\cellcolor[HTML]{CCCCCC} & \cellcolor[HTML]{CCCCCC}       &
	\cellcolor[HTML]{CCCCCC}  &
	\multicolumn{1}{l}{\cellcolor[HTML]{CCCCCC}} &
	\cellcolor[HTML]{CCCCCC} & \cellcolor[HTML]{CCCCCC}       &
	\cellcolor[HTML]{CCCCCC} & \cellcolor[HTML]{CCCCCC}       &
	\cellcolor[HTML]{CCCCCC} & \cellcolor[HTML]{CCCCCC}       &
	\cellcolor[HTML]{CCCCCC} \\
	\multicolumn{1}{c|}{\myelement{4}}      & \multicolumn{1}{c|}{0.25}
	&                                & \multicolumn{1}{l}{}      & 0.37     &
	\checkmark                                   & \multicolumn{1}{c|}{0.37} &
	\multicolumn{1}{c}{\checkmark} & \multicolumn{1}{c}{0.37}
	&                                & \multicolumn{1}{c|}{0.61} &
	\multicolumn{1}{l}{\cellcolor[HTML]{CCCCCC}} &
	\cellcolor[HTML]{CCCCCC} & \cellcolor[HTML]{CCCCCC}       &
	\cellcolor[HTML]{CCCCCC} & \cellcolor[HTML]{CCCCCC}       &
	\cellcolor[HTML]{CCCCCC} & \cellcolor[HTML]{CCCCCC}       &
	\cellcolor[HTML]{CCCCCC} \\
	\multicolumn{1}{c|}{\myelement{5}}      & \multicolumn{1}{c|}{0.25}
	&                                & \multicolumn{1}{l}{}      & 0.37     &
	\checkmark                                   & \multicolumn{1}{c|}{0.37}
	&                                & \multicolumn{1}{c}{0.61} &
	\multicolumn{1}{c}{\checkmark} & \multicolumn{1}{c|}{0.61} &
	\checkmark                                   & \multicolumn{1}{c}{0.61}
	&                                & \multicolumn{1}{c}{1}    &
	\multicolumn{1}{c}{\checkmark} & \multicolumn{1}{c}{1}    &
	\multicolumn{1}{c}{\checkmark} & \multicolumn{1}{c}{1}    \\
	\multicolumn{1}{c|}{\myelement{6}}      & \multicolumn{1}{c|}{0.25}
	&                                & \checkmark                & 0.25     &
	\multicolumn{1}{l}{\cellcolor[HTML]{CCCCCC}} &
	\cellcolor[HTML]{CCCCCC}  & \cellcolor[HTML]{CCCCCC}       &
	\cellcolor[HTML]{CCCCCC} & \cellcolor[HTML]{CCCCCC}       &
	\cellcolor[HTML]{CCCCCC}  &
	\multicolumn{1}{l}{\cellcolor[HTML]{CCCCCC}} &
	\cellcolor[HTML]{CCCCCC} & \cellcolor[HTML]{CCCCCC}       &
	\cellcolor[HTML]{CCCCCC} & \cellcolor[HTML]{CCCCCC}       &
	\cellcolor[HTML]{CCCCCC} & \cellcolor[HTML]{CCCCCC}       &
	\cellcolor[HTML]{CCCCCC} \\
	\multicolumn{1}{c|}{\myelement{7}}      & \multicolumn{1}{c|}{0.25}
	&                                & \checkmark                & 0.25     &
	\multicolumn{1}{l}{\cellcolor[HTML]{CCCCCC}} &
	\cellcolor[HTML]{CCCCCC}  & \cellcolor[HTML]{CCCCCC}       &
	\cellcolor[HTML]{CCCCCC} & \cellcolor[HTML]{CCCCCC}       &
	\cellcolor[HTML]{CCCCCC}  &
	\multicolumn{1}{l}{\cellcolor[HTML]{CCCCCC}} &
	\cellcolor[HTML]{CCCCCC} & \cellcolor[HTML]{CCCCCC}       &
	\cellcolor[HTML]{CCCCCC} & \cellcolor[HTML]{CCCCCC}       &
	\cellcolor[HTML]{CCCCCC} & \cellcolor[HTML]{CCCCCC}       &
	\cellcolor[HTML]{CCCCCC} \\
	\multicolumn{1}{c|}{\myelement{8}}      & \multicolumn{1}{c|}{0.25}
	&                                & \multicolumn{1}{l}{}      & 0.37     &
	\checkmark                                   & \multicolumn{1}{c|}{0.37} &
	\multicolumn{1}{c}{\checkmark} & \multicolumn{1}{c}{0.37}
	&                                & \multicolumn{1}{c|}{0.61} &
	\checkmark                                   & \multicolumn{1}{c}{0.61} &
	\multicolumn{1}{c}{\checkmark} & \multicolumn{1}{c}{0.61} &
	\multicolumn{1}{c}{\checkmark} & \multicolumn{1}{c}{0.61}
	&                                & \multicolumn{1}{c}{1}    \\ \midrule
	\multicolumn{3}{c|}{\property}                                                            &
	\multicolumn{2}{c}{\AlgFalseLiteral} &
	\multicolumn{2}{c|}{\AlgTrueLiteral}                                     &
	\multicolumn{2}{c}{\AlgFalseLiteral}                      &
	\multicolumn{2}{c|}{\AlgFalseLiteral}                      &
	\multicolumn{2}{c}{\AlgTrueLiteral}                                     &
	\multicolumn{2}{c}{\AlgFalseLiteral}                      &
	\multicolumn{2}{c}{\AlgTrueLiteral}                       &
	\multicolumn{2}{c}{\AlgFalseLiteral}                      \\ \bottomrule
\end{tabular}
        }
\end{table}
 \cref{table:running_example_results_probdd}
illustrates the step-by-step results of \probdd.
Following the study of
\probdd~\cite{wang2021probabilistic}, the initial probability \pinit
is set to 0.25, resulting in subsets with a size of 4 as per
\cref{equation:maximize_expected_size_decrease}.
\begin{enumerate}[leftmargin=*]
	\item Round 1 (\mysize{}{}=4).
	Similar to the  example in the
	original paper of
	\probdd~\cite{wang2021probabilistic}, we
	assume \probdd
	selects (\myelement{1}, \myelement{4}, \myelement{5}, \myelement{8})
	to delete due to the randomness, thus
	resulting in the variant \variant{1}.
	However, \variant{1} fails to exhibit \property,
	leading to the probability of these selected elements being updated from
	0.25
	to
	$\frac{0.25}{1
		- (1 - 0.25)^{4}} \approx 0.37$, based on
	\cref{equation:probability_updates}. Next, the remaining elements with
	lower
	probability, \ie, (\myelement{2}, \myelement{3}, \myelement{6},
	\myelement{7}), are prioritized and selected
	for deletion, resulting in \variant{2}. This time, the property test passes
	and these elements are removed.

	\item Round 2 (\mysize{}{}=2).
	Given that all probabilities of remaining elements become 0.37, the
	next subset size
	becomes 2. Subsequently, subset (\myelement{1}, \myelement{5})
	are
	attempted to remove in \variant{3}
	and later subset (\myelement{4}, \myelement{8})
	are attempted to remove in \variant{4}, though no subset can be
	successfully
	removed. After these two attempts, all probabilities
	update to
	$\frac{0.37}{1 - (1 - 0.37)^{2}} \approx 0.61$.

	\item Round 3 (\mysize{}{}=1).
	Finally, the subset size becomes 1,
	so each individual element is selected to remove alone. The elements
	\myelement{4} and \myelement{1} are finally
	removed from the final result in \variant{5} and \variant{7}, respectively,
	while
	\myelement{5} and \myelement{8} are verified as non-removable, thus
	being returned
	as the
	final result.
\end{enumerate}

\section{Delving Deeper into Probability and Size}
\label{sec:finding_probability_and_size}

Beginning with this section, we will systematically present our findings.
Each finding will be introduced by first stating the result, followed by the
explanation.
In this section, we theoretically analyze the trend of
probability changes
across rounds, and the approach to derive the optimal subset size.

\subsection{On the Probability in \probdd}
\label{subsec:finding_probability}



\finding{
The probability assigned to each element increases monotonically with the
round number \RoundNumber,
by a factor of approximately \pincreaseratevalue.
Essentially, the probability for each element can be expressed as a function of \RoundNumber and \pinit, i.e.,
\[ \myprobability{r}{} \approx \pincreaseratevalue^{\RoundNumber} \times \myprobability{0}{} \]
}

\myparagraph{An Illustrative Example}
The running example illustrated in
\cref{table:running_example_results_probdd} leads to this finding.
Observation
reveals that
after each element has been attempted for deletion once, \ie, completing
one
round, the probabilities of all remaining elements are updated. The initial
probability is 0.25; after \variant{2}, it changes to 0.37; following
\variant{4}, it
increases
to  0.61; and by the end of \variant{8}, it reaches 1. Consequently, we
hypothesize that
with each deletion attempt, the probability approximately
increases in a predictable manner. Through appropriate simplification, we
can
theoretically model this trend, and thereby model the entire progression of
probability changes.

\subsubsection{Assumption for Theoretical  Analysis}
\label{subsection:assumption}
Besides the above observation from a concrete example,
theoretical analysis is necessary.
To refine the mathematical
model of \probdd for easier representation, analysis and derivation, we assume that
the number of elements in \listinput is always divisible by the subset size.
With this assumption,
the probability of each element will be
updated in the same manner;
as a result, before and after each round,
the probabilities of all elements are always the same, as shown in
\cref{lem:s:uniform-prob}.
This assumption is often applicable in practice.
For instance, in the running example in
\cref{table:running_example_results_probdd}, before each round, the
probabilities associated with each remaining element are identical, ensuring
that all subsets are of identical size. Furthermore, the probabilities of
elements are updated to the same next value after the round.

\begin{lem}
	If the number of elements in \listinput is always divisible by the subset
	size, then after each round, the probabilities of all elements will always
	remain the same.
	\label{lem:s:uniform-prob}
\end{lem}

%

\begin{proof}
We use	mathematical induction to prove this lemma.

\myparagraph{\underline{\normalfont Base Case}}
Initially, all probabilities are set to the same value. Hence,
before the first round, the probabilities of all elements are identical.

\myparagraph{\underline{\normalfont Inductive Step}}
Assume that before a given round, the probabilities of all
elements are identical (induction hypothesis).
After failing to delete a subset \mysubset, \probdd updates the probability
of each element of \mysubset according to \cref{equation:probability_updates}.
This formula depends solely on two factors: the current probability of each
element of \mysubset, \ie,
\myelement{i}.\myprobability{}{}, and the size of the subset $|\mysubset|$.
For \myelement{i}.\myprobability{}{}, by the induction hypothesis, all
elements have the same probability at the
beginning of the round; for $|\mysubset|$, if the total number of elements
in
\listinput is divisible by the subset size, then every subset in the round will
have the same size $|\mysubset|$.
Therefore, both factors \myelement{i}.\myprobability{}{} and
$|\mysubset|$ are identical for all elements in a subset, and the probabilities
of these elements are updated to the same new value using
\cref{equation:probability_updates}. Furthermore, as all subsets undergo
the same update process, the probabilities of all
elements in the list will
remain identical at the end of the round.

\myparagraph{\underline{\normalfont Conclusion}}
The probabilities of all elements remain identical at the end of this round.
\end{proof}

Consequently, as long as the total number of elements is always divisible by
the subset size, the probabilities of all elements will remain identical
throughout the process.
Take the running example in \cref{table:running_example_results_probdd} as
a demonstration. During the reduction, the number of elements is always
divisible by the subset size in each round, \ie, \mysize{}{}=4, \mysize{}{}=2,
\mysize{}{}=1. Therefore, starting with
an
initial probability of 0.25, the probability of each elements remain identical
after each round, being  0.37, 0.61 and 1, respectively.

While it is not always possible for the number of elements to be divisible by
the subset size, the elements will still be partitioned as evenly as possible.
However, such indivisibilities make the theoretical simplification of \probdd
nearly impossible.
Based on our observation when running \probdd, being slightly uneven
during partitioning does not significantly affect
probability updates.
Moreover, we will demonstrate
that the simplified algorithm derived from this assumption has no
significant difference from \probdd in \cref{sec:implication}, via
thorough experimental evaluation.


\subsubsection{Probability \vs  Subset Size Correlation}

In the second step, we derive the correlation between probability
and subset size. Based on the assumption in the previous step, the
probability of each element is identical and represented as
\myprobability{r}{} in round \RoundNumber, thus the
formula of \gain from  \cref{equation:expected_size_decrease} can be
simplified
to
 \begin{align}
	E(\mysize{}{}) = \mysize{}{} \times (1 -
	\myprobability{r}{})^{\mysize{}{}}
	\label{equation:expected_size_decrease_simplified}
\end{align}
Given the
probability of elements $\myprobability{r}{}$ in the
round \RoundNumber, $\mysize{r}{}$
 can be derived
through gradient-based optimization, \ie,
$E'(\mysize{r}{}) = 0$. Therefore, the
optimal size $\mysize{r}{}$ to
maximize $E(s)$ is $-
\frac{1}{\ln(1-\myprobability{r}{})}$.
Subsequently, we can also deduce the next probability
to be
$\myprobability{r+1}{} = \frac{\myprobability{r}{}}{1 -
(1-\myprobability{r}{})^{\mysize{r}{}}}$.
In summary, the correlation between probability and subset size can be
simplified as  \cref{equation:s_i} and  \cref{equation:p_i+1}, in which
subset size $\mysize{r}{}$ is determined by probability
$\myprobability{r}{}$, and
probability \myprobability{r+1}{} in the next round is
determined by both $\myprobability{r}{}$
and $\mysize{r}{}$.

\begin{empheq}[left=\empheqlbrace]{align}
	\mysize{r}{} &= -\frac{1}{\ln(1 - \myprobability{r}{})}
	\label{equation:s_i} \\
	\myprobability{r+1}{} &= \frac{\myprobability{r}{}}{1 -
	(1-\myprobability{r}{})^{\mysize{r}{}}}
	\label{equation:p_i+1}
\end{empheq}


\subsubsection{Trend of Probability Changes}

Through  \cref{equation:p_i+1}, $\myprobability{r+1}{} > \myprobability{r}{}$ always holds,
indicating a monotonic increase of the probability of elements. However,
there
is still room for simplification, as $\mysize{r}{}$ can be represented by $\myprobability{r}{}$,
implying that \myprobability{r+1}{} can be represented solely by
$\myprobability{r}{}$.

\begin{lem}
    $p$ is increased by a factor $\frac{1}{1 - e^{-1}}$, \ie,
    \begin{align}
        \myprobability{r}{} = \frac{\myprobability{r-1}{}}{1 - e^{-1}} =
        \frac{\pinit}{(1 - e^{-1})^{r}} \label{equation:p_increase}
    \end{align}
    \label{lem:p:increase}
\end{lem}

\begin{proof}
	Given $\mysize{r}{} = - \frac{1}{\ln(1 - \myprobability{r}{})}$, we can
	deduce that $1 - \myprobability{r}{} =
	e^{-\frac{1}{\mysize{r}{}}}$.

	Subsequently, we substitute $1 - \myprobability{r}{}$ into
	\cref{equation:p_i+1},
	obtaining
	$$\myprobability{r+1}{} = \frac{\myprobability{r}{}}{1 -
	(1-\myprobability{r}{})^{\mysize{r}{}}} = \frac{\myprobability{r}{}}{1 -
		(e^{-\frac{1}{\mysize{r}{}}})^{\mysize{r}{}}}$$
		$$
		= \frac{\myprobability{r}{}}{1 - e^{-1}} \approx
	\pincreaseratevalue \times \myprobability{r}{}$$

	Equivalently, the approximate probability
	after round \RoundNumber can be derived given only
	\myprobability{0}{}, \ie,
	$$\myprobability{r}{}  = \frac{\myprobability{0}{}}{(1 - e^{-1})^{r}} \approx
	\pincreaseratevalue^{r} \times \myprobability{0}{}$$
\end{proof}

Therefore, through empirical
observations on the running example, coupled
with theoretical derivation and simplification, we have identified the pattern
of probability changes \wrt the round number \RoundNumber, \ie,
$\myprobability{r}{} = \frac{\pinit}{(1 - e^{-1})^{r}} \approx
\pincreaseratevalue^{\RoundNumber} \times \myprobability{0}{}$.

\subsection{On the Size of Subsets in \probdd}
\label{subsec:finding_size}

\finding{
The size of subsets in $r$-th round can be analytically pre-determined
	given only the
	probability of this round, \ie,
	$\mysize{r}{} = \operatorname*{arg\,max}_{\mysize{}{} \in
		\mathbb{N}^{+}} \mysize{}{} \times (1 -
		\myprobability{r}{})^{\mysize{}{}}
	$, which is either $\lfloor-\frac{1}{\ln(1 -
		\myprobability{r}{})}\rfloor$ or $\lceil-\frac{1}{\ln(1 -
		\myprobability{r}{})}\rceil$.
}

Based on the Finding 1,
the probability \myprobability{r}{}  can be approximately
estimated by the
current round number \RoundNumber via a
factor. Consequently, we can further derive the subset size \mysize{r}{}
by
maximizing the
\gain in \probdd.

\begin{lem}
	The optimal subset size $\mysize{r}{}$ in round $r$
    is either $\lfloor-\frac{1}{\ln(1 -
	\myprobability{r}{})}\rfloor$ or $\lceil-\frac{1}{\ln(1 -
	\myprobability{r}{})}\rceil$.
	\label{lem:s:relation}
\end{lem}

\begin{proof}
	The \gain is determined by the formula $E(\mysize{r}{})=\mysize{r}{}
	\times (1 -
	\myprobability{r}{})^{\mysize{r}{}}$, which increases initially with
	\mysize{r}{} and then decreases as
	\mysize{r}{} grows further, enabling the optimal solution to be identified
	through derivative analysis.
	Therefore, we can deduce the optimal \mysize{r}{} by solving
	$E'(\mysize{r}{}) = 0$.
	Therefore, the optimal size of subsets $\mysize{r}{}$ in r-th round is
	$-\frac{1}{\ln(1
	- \myprobability{r}{})}$, which will be rounded to either
	\vspace{1pt}
	$\lfloor-\frac{1}{\ln(1 -
	\myprobability{r}{})}\rfloor$ \vspace{1pt}
	 or $\lceil-\frac{1}{\ln(1 -
	\myprobability{r}{})}\rceil$.
	\vspace{2pt}
	The final subset size should be chosen based on which
	integer results in a larger \gain.
\end{proof}


\cref{lem:s:relation}
allows the subset size to be analytically
pre-determined, thus providing the potential for simplification of
\probdd and leading to the proposal of \proj (detailed in
\cref{sec:implication}).

\section{Empirical Experiments}
\label{sec:experimental_results}
In addition to the theoretical derivation above, we conduct an extensive
experimental evaluation on \ddmin and \probdd to gain deeper insights and
achieve further discoveries.
Specifically, we reproduce the experiments on \ddmin and \probdd by Wang
\etal \cite{wang2021probabilistic}, and then
delve deeper into \probdd, analyzing its randomness, the bottlenecks it
overcomes, and its 1-minimality.
Furthermore, we evaluate our proposed \proj (which will be presented in
\cref{sec:implication}), validating our previous theoretical analysis.
Due to limited space, we present the results of both \probdd and \proj
together within this section, but this section primarily focuses
on discussing \probdd, while the next section will focus on \proj.

\subsection{Benchmarks}
To extensively evaluate \ddmin, \probdd and \proj, we use the following
three
benchmark suites (\benchmarkNumTotal benchmarks in total), covering
various use scenarios of minimization algorithms.

\begin{itemize}[leftmargin=*]
	\item \cBenchmarkNameShort:
	\cBenchmarkNum large
	bug-triggering
	programs in C language, each of which triggers a real-world compiler
	bug in either \llvm or \gcc. The original size of benchmarks ranges from
	\cBenchmarkSmallestSize
	tokens to \cBenchmarkLargestSize tokens.
	This benchmark suite has been used
	to evaluate test input minimization work~\cite{sun2018perses,
		wang2021probabilistic, vulcan,trec}.

	\item \debloatBenchmarkNameShort:
	source programs of \debloatBenchmarkNum command-line utilities.
	The original size of benchmarks ranges from
	\debloatBenchmarkSmallestSize tokens to
	\debloatBenchmarkLargestSize
	tokens.
	This benchmark suite was collected by  Heo
	\etal~\cite{heo2018effective} and used
	to evaluate software debloating techniques~\cite{heo2018effective,
		qian2019razor, alhanahnah2022lightweight}.

	\item \xmlBenchmarkNameShort:
	\xmlBenchmarkNum \xml inputs triggering \xmlBenchmarkUniqueNum
	unique
	bugs in \basex, a widely-used \xml processing tool.
	The original size of benchmarks ranges from
	\xmlBenchmarkSmallestSize tokens to \xmlBenchmarkLargestSize
	tokens.
	This benchmark suite
	is generated via \xpress~\cite{li2024finding} and collected by the
	authors of this study, as the original \xml dataset used in \probdd
	paper is not publicly available.


\end{itemize}


\subsection{Evaluation Metrics}
We measure the following aspects as metrics.

\myparagraph{Final Size} This metric assesses the effectiveness of
reduction.
When reducing a list \listinput with a certain property \property, a smaller
final list is
preferred, indicating that
more irrelevant elements have been successfully
eliminated.
In all benchmark suites, the metric is measured by \emph{the number of
	tokens}.

\myparagraph{Execution Time}
The execution time of a minimization algorithm reflects its efficiency.
A minimization algorithm taking less time is more desirable, and execution
time is measured in \emph{seconds}.

\myparagraph{Query Number} This metric further evaluates the
algorithm's
efficiency. During the
reduction process, each time a variant
is produced, the algorithm verifies whether this variant still preserves
the property \property,
referred to as a query. Since queries consume time, a
lower query number is favorable.

\myparagraph{P-value} We calculate the p-value via Wilcoxon signed-rank
test\cite{woolson2005wilcoxon}
between every two algorithms, to investigate whether the performance
differences are significant.
A p-value below 0.05 indicates that we can reject the null hypothesis
(which assumes no difference in performance) at the 0.05 significance
level.
Otherwise, we fail to reject the null hypothesis, suggesting that the
observed difference may not be statistically significant.

\subsection{The Wrapping Frameworks}
The \ddmin algorithm and its variants usually serve as the fundamental
algorithm.
To apply them to a concrete scenario,
an outer wrapping framework is generally needed to handle the structure of
the input.
In our evaluation,
we choose the same wrapping frameworks as those used by \probdd paper.
For those tree-structured bug-triggering
inputs, \ie, \cBenchmarkNameShort and \xmlBenchmarkNameShort, we use
\picireny 21.8~\cite{picireny},
an implementation of \hdd~\cite{misherghi2006hdd}. \picireny parses such
inputs into trees, and then invokes \picire 21.8~\cite{picire}, an
open-sourced \deltadebugging library with \ddmin, \probdd and \proj
implemented, to reduce each level of the trees.
For software debloating on \debloatBenchmarkNameShort,
\chisel~\cite{heo2018effective} is employed, in which \ddmin, \probdd and
\proj are integrated.

All experiments are conducted on a server running Ubuntu 22.04.3 LTS,
 with 4 TB RAM
and two Intel Xeon Gold 6348 CPUs @ 2.60GHz.
To ensure the reproducibility, we employ docker images to
release the source code and the configuration.
Each benchmark is reduced using a single thread.
Following the \probdd paper,
we run each algorithm on each benchmark 5 times and calculate the
geometric average results.

\subsection{Reproduction Study of \probdd}
\label{subsec:reproduction}
\begin{table*}[]
		\centering
		\footnotesize
	\caption{The final size, execution time and query number of
		\ddmin, \probdd and
		\proj on \cBenchmarkNameShort and \debloatBenchmarkNameShort.
		To address significant variations across benchmarks, the geometric
		mean rather than the arithmetic mean is employed, providing a
		smoother measure of the average.
	}
	\label{table:rq1_firsthalf}
	\resizebox{0.8\linewidth}{!}{%
\begin{tabular}{c|l|r|rrr|rrr|rrr}
	\toprule
	\multicolumn{1}{l|}{}                         & \multicolumn{1}{c|}{}
	& \multicolumn{1}{c|}{}                                     & \multicolumn{3}{c|}{Final
	size (\#)}                                                                                              &
	\multicolumn{3}{c|}{Execution time
	(s)}                                                                                           &
	\multicolumn{3}{c}{Query
	number}                                                                                                 \\
	\cmidrule(l){4-12}
	\multicolumn{1}{l|}{\multirow{-2}{*}{}}       &
	\multicolumn{1}{c|}{\multirow{-2}{*}{Benchmark}} &
	\multicolumn{1}{c|}{\multirow{-2}{*}{Original size (\#)}} &
	\multicolumn{1}{c}{\cellcolor[HTML]{CCCCCC}ddmin} &
	\multicolumn{1}{c}{\cellcolor[HTML]{EFEFEF}ProbDD} &
	\multicolumn{1}{c|}{CDD} &
	\multicolumn{1}{c}{\cellcolor[HTML]{CCCCCC}ddmin} &
	\multicolumn{1}{c}{\cellcolor[HTML]{EFEFEF}ProbDD} &
	\multicolumn{1}{c|}{CDD} &
	\multicolumn{1}{c}{\cellcolor[HTML]{CCCCCC}ddmin} &
	\multicolumn{1}{c}{\cellcolor[HTML]{EFEFEF}ProbDD} &
	\multicolumn{1}{c}{CDD} \\ \midrule
	& \llvm-22382                                      &
	9,987                                                     &
	\cellcolor[HTML]{CCCCCC}350                       &
	\cellcolor[HTML]{EFEFEF}353                        & 350                      &
	\cellcolor[HTML]{CCCCCC}1,917                     &
	\cellcolor[HTML]{EFEFEF}1,163                      & 1,005                    &
	\cellcolor[HTML]{CCCCCC}11,388                    &
	\cellcolor[HTML]{EFEFEF}5,973                      & 5,262                   \\
	& \llvm-22704                                      &
	184,444                                                   &
	\cellcolor[HTML]{CCCCCC}786                       &
	\cellcolor[HTML]{EFEFEF}764                        & 745                      &
	\cellcolor[HTML]{CCCCCC}27,924                    &
	\cellcolor[HTML]{EFEFEF}12,418                     & 11,371                   &
	\cellcolor[HTML]{CCCCCC}52,412                    &
	\cellcolor[HTML]{EFEFEF}15,425                     & 14,025                  \\
	& \llvm-23309                                      &
	33,310                                                    &
	\cellcolor[HTML]{CCCCCC}1,316                     &
	\cellcolor[HTML]{EFEFEF}1,338                      & 1,265                    &
	\cellcolor[HTML]{CCCCCC}17,619                    &
	\cellcolor[HTML]{EFEFEF}9,991                      & 10,828                   &
	\cellcolor[HTML]{CCCCCC}55,968                    &
	\cellcolor[HTML]{EFEFEF}19,195                     & 17,953                  \\
	& \llvm-23353                                      &
	30,196                                                    &
	\cellcolor[HTML]{CCCCCC}321                       &
	\cellcolor[HTML]{EFEFEF}336                        & 324                      &
	\cellcolor[HTML]{CCCCCC}3,117                     &
	\cellcolor[HTML]{EFEFEF}1,874                      & 1,400                    &
	\cellcolor[HTML]{CCCCCC}11,719                    &
	\cellcolor[HTML]{EFEFEF}5,757                      & 4,492                   \\
	& \llvm-25900                                      &
	78,960                                                    &
	\cellcolor[HTML]{CCCCCC}941                       &
	\cellcolor[HTML]{EFEFEF}932                        & 937                      &
	\cellcolor[HTML]{CCCCCC}7,258                     &
	\cellcolor[HTML]{EFEFEF}3,683                      & 3,104                    &
	\cellcolor[HTML]{CCCCCC}35,740                    &
	\cellcolor[HTML]{EFEFEF}12,553                     & 12,817                  \\
	& \llvm-26760                                      &
	209,577                                                   &
	\cellcolor[HTML]{CCCCCC}520                       &
	\cellcolor[HTML]{EFEFEF}503                        & 498                      &
	\cellcolor[HTML]{CCCCCC}13,123                    &
	\cellcolor[HTML]{EFEFEF}5,876                      & 5,210                    &
	\cellcolor[HTML]{CCCCCC}30,063                    &
	\cellcolor[HTML]{EFEFEF}9,261                      & 9,792                   \\
	& \llvm-27137                                      &
	174,538                                                   &
	\cellcolor[HTML]{CCCCCC}972                       &
	\cellcolor[HTML]{EFEFEF}1,040                      & 966                      &
	\cellcolor[HTML]{CCCCCC}63,971                    &
	\cellcolor[HTML]{EFEFEF}22,208                     & 23,154                   &
	\cellcolor[HTML]{CCCCCC}122,516                   &
	\cellcolor[HTML]{EFEFEF}22,292                     & 20,460                  \\
	& \llvm-27747                                      &
	173,840                                                   &
	\cellcolor[HTML]{CCCCCC}431                       &
	\cellcolor[HTML]{EFEFEF}463                        & 510                      &
	\cellcolor[HTML]{CCCCCC}6,545                     &
	\cellcolor[HTML]{EFEFEF}4,238                      & 2,932                    &
	\cellcolor[HTML]{CCCCCC}20,000                    &
	\cellcolor[HTML]{EFEFEF}8,193                      & 5,992                   \\
	& \llvm-31259                                      &
	48,799                                                    &
	\cellcolor[HTML]{CCCCCC}1,033                     &
	\cellcolor[HTML]{EFEFEF}965                        & 1,035                    &
	\cellcolor[HTML]{CCCCCC}13,815                    &
	\cellcolor[HTML]{EFEFEF}7,497                      & 8,205                    &
	\cellcolor[HTML]{CCCCCC}35,135                    &
	\cellcolor[HTML]{EFEFEF}10,776                     & 13,445                  \\
	& \gcc-59903                                       &
	57,581                                                    &
	\cellcolor[HTML]{CCCCCC}1,185                     &
	\cellcolor[HTML]{EFEFEF}845                        & 743                      &
	\cellcolor[HTML]{CCCCCC}9,067                     &
	\cellcolor[HTML]{EFEFEF}4,879                      & 3,587                    &
	\cellcolor[HTML]{CCCCCC}47,698                    &
	\cellcolor[HTML]{EFEFEF}15,725                     & 13,844                  \\
	& \gcc-60116                                       &
	75,224                                                    &
	\cellcolor[HTML]{CCCCCC}1,615                     &
	\cellcolor[HTML]{EFEFEF}1,628                      & 1,617                    &
	\cellcolor[HTML]{CCCCCC}44,287                    &
	\cellcolor[HTML]{EFEFEF}27,202                     & 27,195                   &
	\cellcolor[HTML]{CCCCCC}80,059                    &
	\cellcolor[HTML]{EFEFEF}27,268                     & 23,204                  \\
	& \gcc-61383                                       &
	32,449                                                    &
	\cellcolor[HTML]{CCCCCC}959                       &
	\cellcolor[HTML]{EFEFEF}966                        & 974                      &
	\cellcolor[HTML]{CCCCCC}12,579                    &
	\cellcolor[HTML]{EFEFEF}6,514                      & 6,566                    &
	\cellcolor[HTML]{CCCCCC}43,716                    &
	\cellcolor[HTML]{EFEFEF}13,593                     & 14,149                  \\
	& \gcc-61917                                       &
	85,359                                                    &
	\cellcolor[HTML]{CCCCCC}882                       &
	\cellcolor[HTML]{EFEFEF}908                        & 884                      &
	\cellcolor[HTML]{CCCCCC}6,740                     &
	\cellcolor[HTML]{EFEFEF}3,591                      & 2,953                    &
	\cellcolor[HTML]{CCCCCC}31,414                    &
	\cellcolor[HTML]{EFEFEF}12,908                     & 14,194                  \\
	& \gcc-64990                                       &
	148,931                                                   &
	\cellcolor[HTML]{CCCCCC}744                       &
	\cellcolor[HTML]{EFEFEF}876                        & 681                      &
	\cellcolor[HTML]{CCCCCC}21,633                    &
	\cellcolor[HTML]{EFEFEF}11,890                     & 11,119                   &
	\cellcolor[HTML]{CCCCCC}44,521                    &
	\cellcolor[HTML]{EFEFEF}16,074                     & 12,112                  \\
	& \gcc-65383                                       &
	43,942                                                    &
	\cellcolor[HTML]{CCCCCC}706                       &
	\cellcolor[HTML]{EFEFEF}701                        & 709                      &
	\cellcolor[HTML]{CCCCCC}5,132                     &
	\cellcolor[HTML]{EFEFEF}3,543                      & 3,358                    &
	\cellcolor[HTML]{CCCCCC}25,051                    &
	\cellcolor[HTML]{EFEFEF}8,591                      & 9,686                   \\
	& \gcc-66186                                       &
	47,481                                                    &
	\cellcolor[HTML]{CCCCCC}1,012                     &
	\cellcolor[HTML]{EFEFEF}981                        & 1,001                    &
	\cellcolor[HTML]{CCCCCC}14,280                    &
	\cellcolor[HTML]{EFEFEF}7,236                      & 12,478                   &
	\cellcolor[HTML]{CCCCCC}47,253                    &
	\cellcolor[HTML]{EFEFEF}12,741                     & 19,094                  \\
	& \gcc-66375                                       &
	65,488                                                    &
	\cellcolor[HTML]{CCCCCC}1,128                     &
	\cellcolor[HTML]{EFEFEF}1,141                      & 1,204                    &
	\cellcolor[HTML]{CCCCCC}23,576                    &
	\cellcolor[HTML]{EFEFEF}14,182                     & 23,229                   &
	\cellcolor[HTML]{CCCCCC}47,339                    &
	\cellcolor[HTML]{EFEFEF}15,690                     & 16,469                  \\
	& \gcc-70127                                       &
	154,816                                                   &
	\cellcolor[HTML]{CCCCCC}934                       &
	\cellcolor[HTML]{EFEFEF}973                        & 930                      &
	\cellcolor[HTML]{CCCCCC}36,390                    &
	\cellcolor[HTML]{EFEFEF}23,925                     & 24,143                   &
	\cellcolor[HTML]{CCCCCC}54,925                    &
	\cellcolor[HTML]{EFEFEF}15,388                     & 15,219                  \\
	& \gcc-70586                                       &
	212,259                                                   &
	\cellcolor[HTML]{CCCCCC}1,583                     &
	\cellcolor[HTML]{EFEFEF}1,561                      & 1,572                    &
	\cellcolor[HTML]{CCCCCC}28,859                    &
	\cellcolor[HTML]{EFEFEF}13,519                     & 15,818                   &
	\cellcolor[HTML]{CCCCCC}102,603                   &
	\cellcolor[HTML]{EFEFEF}24,716                     & 30,715                  \\
	& \gcc-71626                                       &
	4,397                                                     &
	\cellcolor[HTML]{CCCCCC}184                       &
	\cellcolor[HTML]{EFEFEF}184                        & 184                      &
	\cellcolor[HTML]{CCCCCC}119                       &
	\cellcolor[HTML]{EFEFEF}114                        & 99                       &
	\cellcolor[HTML]{CCCCCC}1,608                     &
	\cellcolor[HTML]{EFEFEF}1,156                      & 1,220                   \\
	\cmidrule(l){2-12}
	\multirow{-21}{*}{\cBenchmarkNameShort}       &
	\multicolumn{1}{c|}{Mean}                        &
	64,599                                                    &
	\cellcolor[HTML]{CCCCCC}777                       &
	\cellcolor[HTML]{EFEFEF}775                        & 760                      &
	\cellcolor[HTML]{CCCCCC}10,486                    &
	\cellcolor[HTML]{EFEFEF}5,828                      & 5,676                    &
	\cellcolor[HTML]{CCCCCC}34,013                    &
	\cellcolor[HTML]{EFEFEF}11,652                     & 11,512                  \\ \midrule
	& bzip2-1.0.5                                      &
	70,530                                                    &
	\cellcolor[HTML]{CCCCCC}20,710                    &
	\cellcolor[HTML]{EFEFEF}20,747                     & 20,756                   &
	\cellcolor[HTML]{CCCCCC}137,463                   &
	\cellcolor[HTML]{EFEFEF}105,782                    & 92,058                   &
	\cellcolor[HTML]{CCCCCC}54,034                    &
	\cellcolor[HTML]{EFEFEF}23,240                     & 19,846                  \\
	& chown-8.2                                        &
	43,869                                                    &
	\cellcolor[HTML]{CCCCCC}9,087                     &
	\cellcolor[HTML]{EFEFEF}9,303                      & 9,310                    &
	\cellcolor[HTML]{CCCCCC}38,902                    &
	\cellcolor[HTML]{EFEFEF}25,616                     & 7,625                    &
	\cellcolor[HTML]{CCCCCC}50,487                    &
	\cellcolor[HTML]{EFEFEF}8,278                      & 8,208                   \\
	& date-8.21                                        &
	53,442                                                    &
	\cellcolor[HTML]{CCCCCC}20,604                    &
	\cellcolor[HTML]{EFEFEF}20,738                     & 21,120                   &
	\cellcolor[HTML]{CCCCCC}115,292                   &
	\cellcolor[HTML]{EFEFEF}16,486                     & 15,378                   &
	\cellcolor[HTML]{CCCCCC}139,934                   &
	\cellcolor[HTML]{EFEFEF}14,235                     & 13,928                  \\
	& grep-2.19                                        &
	127,681                                                   &
	\cellcolor[HTML]{CCCCCC}28,723                    &
	\cellcolor[HTML]{EFEFEF}28,627                     & 28,990                   &
	\cellcolor[HTML]{CCCCCC}97,821                    &
	\cellcolor[HTML]{EFEFEF}85,694                     & 93,552                   &
	\cellcolor[HTML]{CCCCCC}277,130                   &
	\cellcolor[HTML]{EFEFEF}42,246                     & 37,607                  \\
	& gzip-1.2.4                                       &
	45,929                                                    &
	\cellcolor[HTML]{CCCCCC}17,065                    &
	\cellcolor[HTML]{EFEFEF}17,068                     & 17,077                   &
	\cellcolor[HTML]{CCCCCC}73,403                    &
	\cellcolor[HTML]{EFEFEF}55,520                     & 75,913                   &
	\cellcolor[HTML]{CCCCCC}147,035                   &
	\cellcolor[HTML]{EFEFEF}27,569                     & 61,032                  \\
	& mkdir-5.2.1                                      &
	34,801                                                    &
	\cellcolor[HTML]{CCCCCC}8,625                     &
	\cellcolor[HTML]{EFEFEF}8,782                      & 8,418                    &
	\cellcolor[HTML]{CCCCCC}3,227                     &
	\cellcolor[HTML]{EFEFEF}2,428                      & 1,877                    &
	\cellcolor[HTML]{CCCCCC}11,969                    &
	\cellcolor[HTML]{EFEFEF}2,836                      & 2,099                   \\
	& rm-8.4                                           &
	44,459                                                    &
	\cellcolor[HTML]{CCCCCC}8,507                     &
	\cellcolor[HTML]{EFEFEF}8,467                      & 8,461                    &
	\cellcolor[HTML]{CCCCCC}12,087                    &
	\cellcolor[HTML]{EFEFEF}5,008                      & 5,109                    &
	\cellcolor[HTML]{CCCCCC}33,171                    &
	\cellcolor[HTML]{EFEFEF}5,097                      & 5,057                   \\
	& sort-8.16                                        &
	88,068                                                    &
	\cellcolor[HTML]{CCCCCC}14,893                    &
	\cellcolor[HTML]{EFEFEF}14,843                     & 15,834                   &
	\cellcolor[HTML]{CCCCCC}60,631                    &
	\cellcolor[HTML]{EFEFEF}61,739                     & 21,948                   &
	\cellcolor[HTML]{CCCCCC}119,150                   &
	\cellcolor[HTML]{EFEFEF}18,711                     & 7,914                   \\
	& tar-1.14                                         &
	163,296                                                   &
	\cellcolor[HTML]{CCCCCC}20,411                    &
	\cellcolor[HTML]{EFEFEF}20,713                     & 20,592                   &
	\cellcolor[HTML]{CCCCCC}115,234                   &
	\cellcolor[HTML]{EFEFEF}95,765                     & 77,910                   &
	\cellcolor[HTML]{CCCCCC}200,394                   &
	\cellcolor[HTML]{EFEFEF}14,384                     & 12,095                  \\
	& uniq-8.16                                        &
	63,861                                                    &
	\cellcolor[HTML]{CCCCCC}14,350                    &
	\cellcolor[HTML]{EFEFEF}14,262                     & 14,354                   &
	\cellcolor[HTML]{CCCCCC}21,672                    &
	\cellcolor[HTML]{EFEFEF}23,177                     & 19,124                   &
	\cellcolor[HTML]{CCCCCC}25,886                    &
	\cellcolor[HTML]{EFEFEF}4,228                      & 3,669                   \\
	\cmidrule(l){2-12}
	\multirow{-11}{*}{\debloatBenchmarkNameShort} &
	\multicolumn{1}{c|}{Mean}                        &
	65,151                                                    &
	\cellcolor[HTML]{CCCCCC}15,080                    &
	\cellcolor[HTML]{EFEFEF}15,152                     & 15,235                   &
	\cellcolor[HTML]{CCCCCC}43,827                    &
	\cellcolor[HTML]{EFEFEF}28,505                     & 21,782                   &
	\cellcolor[HTML]{CCCCCC}72,140                    &
	\cellcolor[HTML]{EFEFEF}11,803                     & 10,686                  \\
	\bottomrule
\end{tabular}
}
\end{table*}
\begin{table*}[]
	\centering
	\footnotesize
	\caption{The final size, execution time and query number of
		\ddmin, \probdd and
		\proj on \xmlBenchmarkNameShort.
		The last row shows the overall average across all
		three
		benchmark suites.
	}
	\label{table:rq1_secondhalf}
	\resizebox{0.8\linewidth}{!}{%
\begin{tabular}{c|l|r|rrr|rrr|rrr}
	\toprule
	\multicolumn{1}{l|}{}                     & \multicolumn{1}{c|}{}                            &
	\multicolumn{1}{c|}{}                                     & \multicolumn{3}{c|}{Final
	size (\#)}                                                                                              &
	\multicolumn{3}{c|}{Execution time
	(s)}                                                                                           &
	\multicolumn{3}{c}{Query
	number}                                                                                                 \\
	\cmidrule(l){4-12}
	\multicolumn{1}{l|}{\multirow{-2}{*}{}}   &
	\multicolumn{1}{c|}{\multirow{-2}{*}{Benchmark}} &
	\multicolumn{1}{c|}{\multirow{-2}{*}{Original size (\#)}} &
	\multicolumn{1}{c}{\cellcolor[HTML]{CCCCCC}ddmin} &
	\multicolumn{1}{c}{\cellcolor[HTML]{EFEFEF}ProbDD} &
	\multicolumn{1}{c|}{CDD} &
	\multicolumn{1}{c}{\cellcolor[HTML]{CCCCCC}ddmin} &
	\multicolumn{1}{c}{\cellcolor[HTML]{EFEFEF}ProbDD} &
	\multicolumn{1}{c|}{CDD} &
	\multicolumn{1}{c}{\cellcolor[HTML]{CCCCCC}ddmin} &
	\multicolumn{1}{c}{\cellcolor[HTML]{EFEFEF}ProbDD} &
	\multicolumn{1}{c}{CDD} \\ \midrule
	& xml-071d221-1                                    &
	20,090                                                    &
	\cellcolor[HTML]{CCCCCC}10                        &
	\cellcolor[HTML]{EFEFEF}15                         & 20                       &
	\cellcolor[HTML]{CCCCCC}73                        &
	\cellcolor[HTML]{EFEFEF}114                        & 144                      &
	\cellcolor[HTML]{CCCCCC}29                        &
	\cellcolor[HTML]{EFEFEF}50                         & 69                      \\
	& xml-071d221-2                                    &
	20,387                                                    &
	\cellcolor[HTML]{CCCCCC}13                        &
	\cellcolor[HTML]{EFEFEF}14                         & 20                       &
	\cellcolor[HTML]{CCCCCC}155                       &
	\cellcolor[HTML]{EFEFEF}146                        & 243                      &
	\cellcolor[HTML]{CCCCCC}60                        &
	\cellcolor[HTML]{EFEFEF}69                         & 117                     \\
	& xml-1e9bc83-1                                    &
	20,327                                                    &
	\cellcolor[HTML]{CCCCCC}38                        &
	\cellcolor[HTML]{EFEFEF}49                         & 24                       &
	\cellcolor[HTML]{CCCCCC}491                       &
	\cellcolor[HTML]{EFEFEF}522                        & 284                      &
	\cellcolor[HTML]{CCCCCC}235                       &
	\cellcolor[HTML]{EFEFEF}236                        & 115                     \\
	& xml-1e9bc83-2                                    &
	20,222                                                    &
	\cellcolor[HTML]{CCCCCC}79                        &
	\cellcolor[HTML]{EFEFEF}78                         & 76                       &
	\cellcolor[HTML]{CCCCCC}1,391                     &
	\cellcolor[HTML]{EFEFEF}814                        & 867                      &
	\cellcolor[HTML]{CCCCCC}725                       &
	\cellcolor[HTML]{EFEFEF}390                        & 384                     \\
	& xml-1e9bc83-3                                    &
	20,219                                                    &
	\cellcolor[HTML]{CCCCCC}69                        &
	\cellcolor[HTML]{EFEFEF}70                         & 72                       &
	\cellcolor[HTML]{CCCCCC}1,313                     &
	\cellcolor[HTML]{EFEFEF}893                        & 889                      &
	\cellcolor[HTML]{CCCCCC}619                       &
	\cellcolor[HTML]{EFEFEF}416                        & 404                     \\
	& xml-1e9bc83-4                                    &
	19,985                                                    &
	\cellcolor[HTML]{CCCCCC}156                       &
	\cellcolor[HTML]{EFEFEF}139                        & 143                      &
	\cellcolor[HTML]{CCCCCC}3,911                     &
	\cellcolor[HTML]{EFEFEF}1,939                      & 2,521                    &
	\cellcolor[HTML]{CCCCCC}1,943                     &
	\cellcolor[HTML]{EFEFEF}935                        & 1,173                   \\
	& xml-1e9bc83-5                                    &
	20,579                                                    &
	\cellcolor[HTML]{CCCCCC}81                        &
	\cellcolor[HTML]{EFEFEF}73                         & 75                       &
	\cellcolor[HTML]{CCCCCC}1,355                     &
	\cellcolor[HTML]{EFEFEF}929                        & 882                      &
	\cellcolor[HTML]{CCCCCC}746                       &
	\cellcolor[HTML]{EFEFEF}485                        & 428                     \\
	& xml-1e9bc83-6                                    &
	19,880                                                    &
	\cellcolor[HTML]{CCCCCC}127                       &
	\cellcolor[HTML]{EFEFEF}126                        & 124                      &
	\cellcolor[HTML]{CCCCCC}3,563                     &
	\cellcolor[HTML]{EFEFEF}1,852                      & 1,548                    &
	\cellcolor[HTML]{CCCCCC}1,907                     &
	\cellcolor[HTML]{EFEFEF}964                        & 749                     \\
	& xml-1e9bc83-7                                    &
	20,297                                                    &
	\cellcolor[HTML]{CCCCCC}111                       &
	\cellcolor[HTML]{EFEFEF}114                        & 111                      &
	\cellcolor[HTML]{CCCCCC}3,419                     &
	\cellcolor[HTML]{EFEFEF}1,684                      & 2,330                    &
	\cellcolor[HTML]{CCCCCC}1,757                     &
	\cellcolor[HTML]{EFEFEF}827                        & 931                     \\
	& xml-1e9bc83-8                                    &
	20,327                                                    &
	\cellcolor[HTML]{CCCCCC}100                       &
	\cellcolor[HTML]{EFEFEF}107                        & 100                      &
	\cellcolor[HTML]{CCCCCC}2,862                     &
	\cellcolor[HTML]{EFEFEF}1,636                      & 1,446                    &
	\cellcolor[HTML]{CCCCCC}1,451                     &
	\cellcolor[HTML]{EFEFEF}751                        & 592                     \\
	& xml-1e9bc83-9                                    &
	20,330                                                    &
	\cellcolor[HTML]{CCCCCC}128                       &
	\cellcolor[HTML]{EFEFEF}73                         & 128                      &
	\cellcolor[HTML]{CCCCCC}2,850                     &
	\cellcolor[HTML]{EFEFEF}1,230                      & 2,494                    &
	\cellcolor[HTML]{CCCCCC}1,437                     &
	\cellcolor[HTML]{EFEFEF}561                        & 832                     \\
	& xml-2d4ec80-1                                    &
	20,129                                                    &
	\cellcolor[HTML]{CCCCCC}76                        &
	\cellcolor[HTML]{EFEFEF}72                         & 76                       &
	\cellcolor[HTML]{CCCCCC}570                       &
	\cellcolor[HTML]{EFEFEF}522                        & 791                      &
	\cellcolor[HTML]{CCCCCC}384                       &
	\cellcolor[HTML]{EFEFEF}304                        & 354                     \\
	& xml-327c8af-1                                    &
	20,207                                                    &
	\cellcolor[HTML]{CCCCCC}55                        &
	\cellcolor[HTML]{EFEFEF}55                         & 55                       &
	\cellcolor[HTML]{CCCCCC}864                       &
	\cellcolor[HTML]{EFEFEF}625                        & 958                      &
	\cellcolor[HTML]{CCCCCC}527                       &
	\cellcolor[HTML]{EFEFEF}319                        & 392                     \\
	& xml-3398ac2-1                                    &
	20,414                                                    &
	\cellcolor[HTML]{CCCCCC}45                        &
	\cellcolor[HTML]{EFEFEF}44                         & 48                       &
	\cellcolor[HTML]{CCCCCC}345                       &
	\cellcolor[HTML]{EFEFEF}383                        & 605                      &
	\cellcolor[HTML]{CCCCCC}225                       &
	\cellcolor[HTML]{EFEFEF}192                        & 268                     \\
	& xml-3398ac2-2                                    &
	19,290                                                    &
	\cellcolor[HTML]{CCCCCC}48                        &
	\cellcolor[HTML]{EFEFEF}48                         & 48                       &
	\cellcolor[HTML]{CCCCCC}613                       &
	\cellcolor[HTML]{EFEFEF}492                        & 578                      &
	\cellcolor[HTML]{CCCCCC}358                       &
	\cellcolor[HTML]{EFEFEF}255                        & 238                     \\
	& xml-3398ac2-3                                    &
	20,222                                                    &
	\cellcolor[HTML]{CCCCCC}62                        &
	\cellcolor[HTML]{EFEFEF}62                         & 62                       &
	\cellcolor[HTML]{CCCCCC}632                       &
	\cellcolor[HTML]{EFEFEF}574                        & 764                      &
	\cellcolor[HTML]{CCCCCC}347                       &
	\cellcolor[HTML]{EFEFEF}276                        & 306                     \\
	& xml-3398ac2-4                                    &
	19,913                                                    &
	\cellcolor[HTML]{CCCCCC}111                       &
	\cellcolor[HTML]{EFEFEF}112                        & 111                      &
	\cellcolor[HTML]{CCCCCC}1,419                     &
	\cellcolor[HTML]{EFEFEF}1,197                      & 1,722                    &
	\cellcolor[HTML]{CCCCCC}802                       &
	\cellcolor[HTML]{EFEFEF}584                        & 680                     \\
	& xml-3398ac2-5                                    &
	20,477                                                    &
	\cellcolor[HTML]{CCCCCC}44                        &
	\cellcolor[HTML]{EFEFEF}38                         & 44                       &
	\cellcolor[HTML]{CCCCCC}850                       &
	\cellcolor[HTML]{EFEFEF}531                        & 633                      &
	\cellcolor[HTML]{CCCCCC}507                       &
	\cellcolor[HTML]{EFEFEF}269                        & 261                     \\
	& xml-4c99b96-1                                    &
	20,513                                                    &
	\cellcolor[HTML]{CCCCCC}80                        &
	\cellcolor[HTML]{EFEFEF}78                         & 80                       &
	\cellcolor[HTML]{CCCCCC}1,427                     &
	\cellcolor[HTML]{EFEFEF}1,126                      & 1,385                    &
	\cellcolor[HTML]{CCCCCC}699                       &
	\cellcolor[HTML]{EFEFEF}452                        & 438                     \\
	& xml-4c99b96-10                                   &
	20,522                                                    &
	\cellcolor[HTML]{CCCCCC}46                        &
	\cellcolor[HTML]{EFEFEF}47                         & 46                       &
	\cellcolor[HTML]{CCCCCC}1,012                     &
	\cellcolor[HTML]{EFEFEF}760                        & 873                      &
	\cellcolor[HTML]{CCCCCC}443                       &
	\cellcolor[HTML]{EFEFEF}299                        & 256                     \\
	& xml-4c99b96-11                                   &
	19,901                                                    &
	\cellcolor[HTML]{CCCCCC}53                        &
	\cellcolor[HTML]{EFEFEF}54                         & 53                       &
	\cellcolor[HTML]{CCCCCC}868                       &
	\cellcolor[HTML]{EFEFEF}661                        & 791                      &
	\cellcolor[HTML]{CCCCCC}357                       &
	\cellcolor[HTML]{EFEFEF}252                        & 236                     \\
	& xml-4c99b96-12                                   &
	19,775                                                    &
	\cellcolor[HTML]{CCCCCC}102                       &
	\cellcolor[HTML]{EFEFEF}104                        & 102                      &
	\cellcolor[HTML]{CCCCCC}2,429                     &
	\cellcolor[HTML]{EFEFEF}1,592                      & 2,422                    &
	\cellcolor[HTML]{CCCCCC}1,077                     &
	\cellcolor[HTML]{EFEFEF}626                        & 773                     \\
	& xml-4c99b96-13                                   &
	20,114                                                    &
	\cellcolor[HTML]{CCCCCC}60                        &
	\cellcolor[HTML]{EFEFEF}61                         & 60                       &
	\cellcolor[HTML]{CCCCCC}1,132                     &
	\cellcolor[HTML]{EFEFEF}870                        & 989                      &
	\cellcolor[HTML]{CCCCCC}494                       &
	\cellcolor[HTML]{EFEFEF}335                        & 309                     \\
	& xml-4c99b96-14                                   &
	19,970                                                    &
	\cellcolor[HTML]{CCCCCC}102                       &
	\cellcolor[HTML]{EFEFEF}103                        & 102                      &
	\cellcolor[HTML]{CCCCCC}2,147                     &
	\cellcolor[HTML]{EFEFEF}1,687                      & 1,534                    &
	\cellcolor[HTML]{CCCCCC}987                       &
	\cellcolor[HTML]{EFEFEF}600                        & 504                     \\
	& xml-4c99b96-15                                   &
	20,138                                                    &
	\cellcolor[HTML]{CCCCCC}46                        &
	\cellcolor[HTML]{EFEFEF}46                         & 46                       &
	\cellcolor[HTML]{CCCCCC}843                       &
	\cellcolor[HTML]{EFEFEF}788                        & 729                      &
	\cellcolor[HTML]{CCCCCC}381                       &
	\cellcolor[HTML]{EFEFEF}297                        & 244                     \\
	& xml-4c99b96-16                                   &
	20,126                                                    &
	\cellcolor[HTML]{CCCCCC}67                        &
	\cellcolor[HTML]{EFEFEF}68                         & 70                       &
	\cellcolor[HTML]{CCCCCC}1,366                     &
	\cellcolor[HTML]{EFEFEF}1,189                      & 1,125                    &
	\cellcolor[HTML]{CCCCCC}658                       &
	\cellcolor[HTML]{EFEFEF}414                        & 409                     \\
	& xml-4c99b96-17                                   &
	20,210                                                    &
	\cellcolor[HTML]{CCCCCC}67                        &
	\cellcolor[HTML]{EFEFEF}68                         & 70                       &
	\cellcolor[HTML]{CCCCCC}1,360                     &
	\cellcolor[HTML]{EFEFEF}1,310                      & 1,120                    &
	\cellcolor[HTML]{CCCCCC}658                       &
	\cellcolor[HTML]{EFEFEF}414                        & 409                     \\
	& xml-4c99b96-18                                   &
	20,390                                                    &
	\cellcolor[HTML]{CCCCCC}28                        &
	\cellcolor[HTML]{EFEFEF}28                         & 25                       &
	\cellcolor[HTML]{CCCCCC}724                       &
	\cellcolor[HTML]{EFEFEF}492                        & 410                      &
	\cellcolor[HTML]{CCCCCC}324                       &
	\cellcolor[HTML]{EFEFEF}164                        & 137                     \\
	& xml-4c99b96-19                                   &
	20,192                                                    &
	\cellcolor[HTML]{CCCCCC}81                        &
	\cellcolor[HTML]{EFEFEF}82                         & 81                       &
	\cellcolor[HTML]{CCCCCC}1,729                     &
	\cellcolor[HTML]{EFEFEF}1,798                      & 1,895                    &
	\cellcolor[HTML]{CCCCCC}793                       &
	\cellcolor[HTML]{EFEFEF}540                        & 710                     \\
	& xml-4c99b96-2                                    &
	20,750                                                    &
	\cellcolor[HTML]{CCCCCC}53                        &
	\cellcolor[HTML]{EFEFEF}54                         & 53                       &
	\cellcolor[HTML]{CCCCCC}878                       &
	\cellcolor[HTML]{EFEFEF}1,017                      & 912                      &
	\cellcolor[HTML]{CCCCCC}362                       &
	\cellcolor[HTML]{EFEFEF}281                        & 318                     \\
	& xml-4c99b96-3                                    &
	20,015                                                    &
	\cellcolor[HTML]{CCCCCC}67                        &
	\cellcolor[HTML]{EFEFEF}68                         & 67                       &
	\cellcolor[HTML]{CCCCCC}1,102                     &
	\cellcolor[HTML]{EFEFEF}1,254                      & 900                      &
	\cellcolor[HTML]{CCCCCC}483                       &
	\cellcolor[HTML]{EFEFEF}357                        & 337                     \\
	& xml-4c99b96-4                                    &
	20,201                                                    &
	\cellcolor[HTML]{CCCCCC}67                        &
	\cellcolor[HTML]{EFEFEF}67                         & 67                       &
	\cellcolor[HTML]{CCCCCC}1,097                     &
	\cellcolor[HTML]{EFEFEF}1,238                      & 860                      &
	\cellcolor[HTML]{CCCCCC}482                       &
	\cellcolor[HTML]{EFEFEF}365                        & 328                     \\
	& xml-4c99b96-5                                    &
	20,279                                                    &
	\cellcolor[HTML]{CCCCCC}63                        &
	\cellcolor[HTML]{EFEFEF}62                         & 60                       &
	\cellcolor[HTML]{CCCCCC}1,170                     &
	\cellcolor[HTML]{EFEFEF}1,155                      & 1,029                    &
	\cellcolor[HTML]{CCCCCC}515                       &
	\cellcolor[HTML]{EFEFEF}347                        & 364                     \\
	& xml-4c99b96-6                                    &
	19,973                                                    &
	\cellcolor[HTML]{CCCCCC}81                        &
	\cellcolor[HTML]{EFEFEF}82                         & 87                       &
	\cellcolor[HTML]{CCCCCC}1,258                     &
	\cellcolor[HTML]{EFEFEF}1,338                      & 1,056                    &
	\cellcolor[HTML]{CCCCCC}511                       &
	\cellcolor[HTML]{EFEFEF}381                        & 441                     \\
	& xml-4c99b96-7                                    &
	19,973                                                    &
	\cellcolor[HTML]{CCCCCC}115                       &
	\cellcolor[HTML]{EFEFEF}115                        & 115                      &
	\cellcolor[HTML]{CCCCCC}3,399                     &
	\cellcolor[HTML]{EFEFEF}2,560                      & 1,535                    &
	\cellcolor[HTML]{CCCCCC}1,485                     &
	\cellcolor[HTML]{EFEFEF}811                        & 676                     \\
	& xml-4c99b96-8                                    &
	20,579                                                    &
	\cellcolor[HTML]{CCCCCC}42                        &
	\cellcolor[HTML]{EFEFEF}41                         & 42                       &
	\cellcolor[HTML]{CCCCCC}926                       &
	\cellcolor[HTML]{EFEFEF}810                        & 687                      &
	\cellcolor[HTML]{CCCCCC}422                       &
	\cellcolor[HTML]{EFEFEF}254                        & 323                     \\
	& xml-4c99b96-9                                    &
	20,075                                                    &
	\cellcolor[HTML]{CCCCCC}53                        &
	\cellcolor[HTML]{EFEFEF}53                         & 53                       &
	\cellcolor[HTML]{CCCCCC}884                       &
	\cellcolor[HTML]{EFEFEF}940                        & 650                      &
	\cellcolor[HTML]{CCCCCC}368                       &
	\cellcolor[HTML]{EFEFEF}278                        & 265                     \\
	& xml-8ede045-1                                    &
	20,192                                                    &
	\cellcolor[HTML]{CCCCCC}48                        &
	\cellcolor[HTML]{EFEFEF}61                         & 69                       &
	\cellcolor[HTML]{CCCCCC}1,152                     &
	\cellcolor[HTML]{EFEFEF}1,232                      & 1,137                    &
	\cellcolor[HTML]{CCCCCC}499                       &
	\cellcolor[HTML]{EFEFEF}389                        & 487                     \\
	& xml-8ede045-2                                    &
	20,177                                                    &
	\cellcolor[HTML]{CCCCCC}17                        &
	\cellcolor[HTML]{EFEFEF}22                         & 34                       &
	\cellcolor[HTML]{CCCCCC}316                       &
	\cellcolor[HTML]{EFEFEF}335                        & 360                      &
	\cellcolor[HTML]{CCCCCC}132                       &
	\cellcolor[HTML]{EFEFEF}107                        & 164                     \\
	& xml-8ede045-3                                    &
	20,393                                                    &
	\cellcolor[HTML]{CCCCCC}31                        &
	\cellcolor[HTML]{EFEFEF}31                         & 31                       &
	\cellcolor[HTML]{CCCCCC}476                       &
	\cellcolor[HTML]{EFEFEF}431                        & 313                      &
	\cellcolor[HTML]{CCCCCC}199                       &
	\cellcolor[HTML]{EFEFEF}134                        & 134                     \\
	& xml-8ede045-4                                    &
	20,123                                                    &
	\cellcolor[HTML]{CCCCCC}31                        &
	\cellcolor[HTML]{EFEFEF}25                         & 31                       &
	\cellcolor[HTML]{CCCCCC}578                       &
	\cellcolor[HTML]{EFEFEF}472                        & 435                      &
	\cellcolor[HTML]{CCCCCC}250                       &
	\cellcolor[HTML]{EFEFEF}151                        & 181                     \\
	& xml-8ede045-5                                    &
	20,051                                                    &
	\cellcolor[HTML]{CCCCCC}17                        &
	\cellcolor[HTML]{EFEFEF}20                         & 17                       &
	\cellcolor[HTML]{CCCCCC}185                       &
	\cellcolor[HTML]{EFEFEF}246                        & 204                      &
	\cellcolor[HTML]{CCCCCC}61                        &
	\cellcolor[HTML]{EFEFEF}72                         & 80                      \\
	& xml-8ede045-6                                    &
	20,636                                                    &
	\cellcolor[HTML]{CCCCCC}73                        &
	\cellcolor[HTML]{EFEFEF}80                         & 76                       &
	\cellcolor[HTML]{CCCCCC}1,365                     &
	\cellcolor[HTML]{EFEFEF}1,479                      & 1,545                    &
	\cellcolor[HTML]{CCCCCC}552                       &
	\cellcolor[HTML]{EFEFEF}469                        & 700                     \\
	& xml-8ede045-7                                    &
	20,054                                                    &
	\cellcolor[HTML]{CCCCCC}106                       &
	\cellcolor[HTML]{EFEFEF}106                        & 106                      &
	\cellcolor[HTML]{CCCCCC}2,851                     &
	\cellcolor[HTML]{EFEFEF}1,686                      & 1,237                    &
	\cellcolor[HTML]{CCCCCC}1,447                     &
	\cellcolor[HTML]{EFEFEF}614                        & 634                     \\
	& xml-8ede045-8                                    &
	20,177                                                    &
	\cellcolor[HTML]{CCCCCC}76                        &
	\cellcolor[HTML]{EFEFEF}78                         & 76                       &
	\cellcolor[HTML]{CCCCCC}1,397                     &
	\cellcolor[HTML]{EFEFEF}1,327                      & 1,060                    &
	\cellcolor[HTML]{CCCCCC}597                       &
	\cellcolor[HTML]{EFEFEF}449                        & 475                     \\
	& xml-f053486-1                                    &
	20,030                                                    &
	\cellcolor[HTML]{CCCCCC}10                        &
	\cellcolor[HTML]{EFEFEF}10                         & 10                       &
	\cellcolor[HTML]{CCCCCC}101                       &
	\cellcolor[HTML]{EFEFEF}134                        & 76                       &
	\cellcolor[HTML]{CCCCCC}31                        &
	\cellcolor[HTML]{EFEFEF}41                         & 28                      \\
	\cmidrule(l){2-12}
	\multirow{-47}{*}{\xmlBenchmarkNameShort} &
	\multicolumn{1}{c|}{Mean}                        &
	20,190                                                    &
	\cellcolor[HTML]{CCCCCC}56                        &
	\cellcolor[HTML]{EFEFEF}56                         & 58                       &
	\cellcolor[HTML]{CCCCCC}972                       &
	\cellcolor[HTML]{EFEFEF}819                        & 821                      &
	\cellcolor[HTML]{CCCCCC}453                       &
	\cellcolor[HTML]{EFEFEF}314                        & 327                     \\ \midrule
	All                                       & \multicolumn{1}{c|}{Mean}                        &
	31,989                                                    &
	\cellcolor[HTML]{CCCCCC}233                       &
	\cellcolor[HTML]{EFEFEF}235                        & 237                      &
	\cellcolor[HTML]{CCCCCC}2,999                     &
	\cellcolor[HTML]{EFEFEF}2,189                      & 2,102                    &
	\cellcolor[HTML]{CCCCCC}2,752                     &
	\cellcolor[HTML]{EFEFEF}1,309                      & 1,320                   \\
	\bottomrule
\end{tabular}
	}
\end{table*}

To comprehensively reproduce the results of
\probdd~\cite{wang2021probabilistic}, we evaluate \ddmin and \probdd
using three benchmark suites, containing a total of \benchmarkNumTotal
benchmarks.
Following the settings of \probdd~\cite{wang2021probabilistic},
we set the empirically estimated remaining rate as the initialization
probability
\pinit, specifically, 0.1 for \cBenchmarkNameShort and
\debloatBenchmarkNameShort, and 2.5e-3 for \xmlBenchmarkNameShort.
The detailed results are shown in
\cref{table:rq1_firsthalf} and \cref{table:rq1_secondhalf}.



\myparagraph{Efficiency and Effectiveness}
Through our
reproduction study, we find that the performance of \probdd aligns with the
results reported in the original paper, showing that \probdd is
significantly more efficient than \ddmin.
Across three benchmark suites, \probdd
requires \allBenchmarkNoTimeoutTimeDecRateProbddVsDdmin less time
and
\allBenchmarkNoTimeoutQueryDecRateProbddVsDdmin fewer queries,
with p-value being
\allBenchmarkTimeProbddVsDdminPvalue and
\allBenchmarkQueryProbddVsDdminPvalue, respectively.
Moreover, we
assess the effectiveness by measuring the sizes of the final minimized
results. The effectiveness of \ddmin and \probdd varies across
benchmarks. However, the Wilcoxon signed-rank test yields a p-value of
\allBenchmarkSizeProbddVsDdminPvalue, which is substantially higher
than the 0.05 significance level. Therefore, we fail to reject the null
hypothesis, suggesting no statistically significant difference in overall
performance between the two algorithms.






\subsection{Impact of Randomness in \probdd}
\label{subsec:finding_randomness}
\finding{Randomness has no significant
	impact on the performance of \probdd.
}

In \probdd, elements with different probabilities are sorted accordingly,
while elements with the same probability are randomly shuffled.
However, randomness alone intuitively does not ensure a higher probability
of
escaping local optima and
the effect
of this randomness on performance has not been thoroughly
investigated.

To this end, we conduct an ablation study by removing
such randomness, creating a variant called \probddnor. We evaluate this
variant across all benchmarks.
The results indicate that the randomness
does not significantly impact
performance. Specifically, in terms of final
size,
execution time, and query number, \probddnor achieves
\allBenchmarkSizeAvgProbddNoshuffle,
\allBenchmarkTimeAvgProbddNoshuffle, and
\allBenchmarkQueryAvgProbddNoshuffle compared to
\allBenchmarkSizeAvgProbdd, \allBenchmarkTimeAvgProbdd, and
\allBenchmarkQueryAvgProbdd of \probdd, respectively.
The p-values of \allBenchmarkSizeProbddVsProbddNoshufflePvalue,
\allBenchmarkTimeProbddVsProbddNoshufflePvalue, and
\allBenchmarkQueryProbddVsProbddNoshufflePvalue all exceed 0.05, so
we fail to reject the null hypothesis, indicating no statistically significant
differences.

\subsection{Bottleneck Overcome by \probdd}
\label{subsec:finding_bottleneck}
\finding{On tree-structured inputs, inefficient deletion
	attempts on complements
	and repeated
	attempts account for the bottlenecks of \ddmin, which are overcome by
	\probdd.
}
In the study of \probdd, the authors demonstrate that \probdd is
more efficient than the baseline approach (\ddmin) in tree-based reduction
scenarios, where the inputs are parsed into tree representations before
reduction. Therefore, to uncover the root cause of this superiority,
we follow the same application scenario and analyze the behavior of
\probdd
in reducing the tree-structured inputs.

\begin{figure*}[htbp]
	\centering
	\begin{subfigure}{0.32\linewidth}
		\centering
		\includegraphics[width=\linewidth]{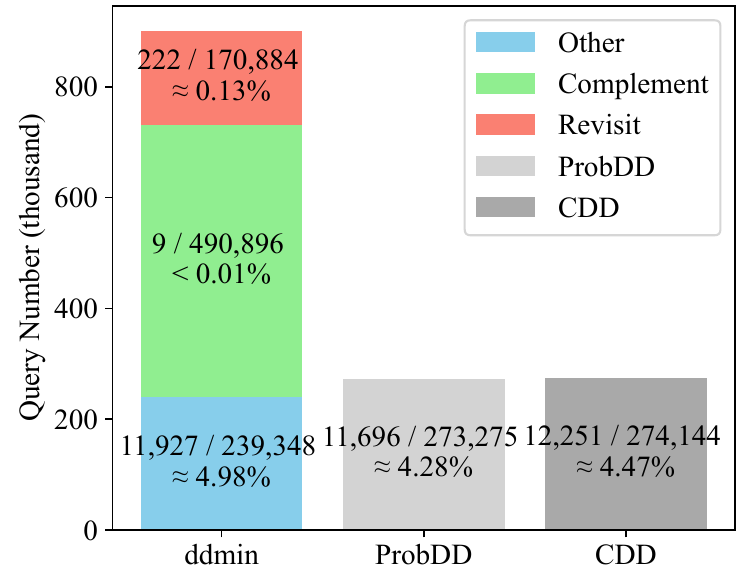}
		\caption{ On \cBenchmarkNameShort }
		\label{fig:c_stacked_barchart}
	\end{subfigure}
	\hfil
	\begin{subfigure}{0.32\linewidth}
		\centering
		\includegraphics[width=\linewidth]{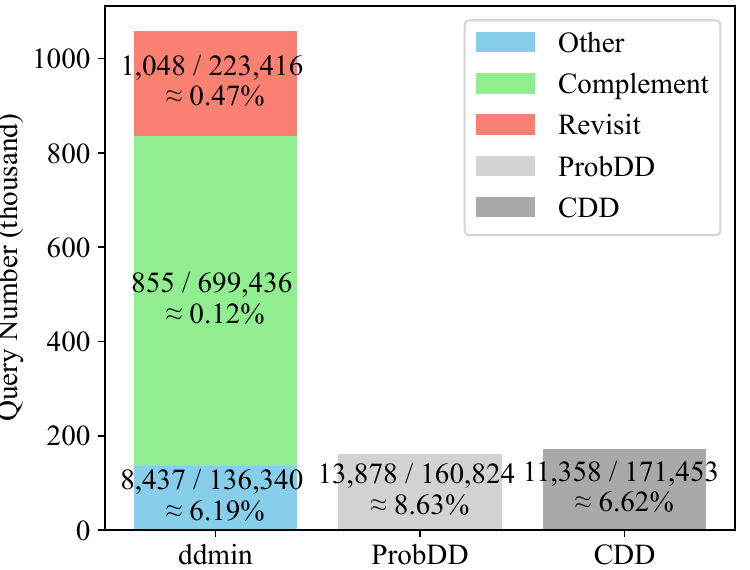}
		\caption{ On \debloatBenchmarkNameShort }
		\label{fig:chisel_stacked_barchart}
	\end{subfigure}
    \hfil
     \begin{subfigure}{0.32\linewidth}
    	\centering
    	\includegraphics[width=\linewidth]{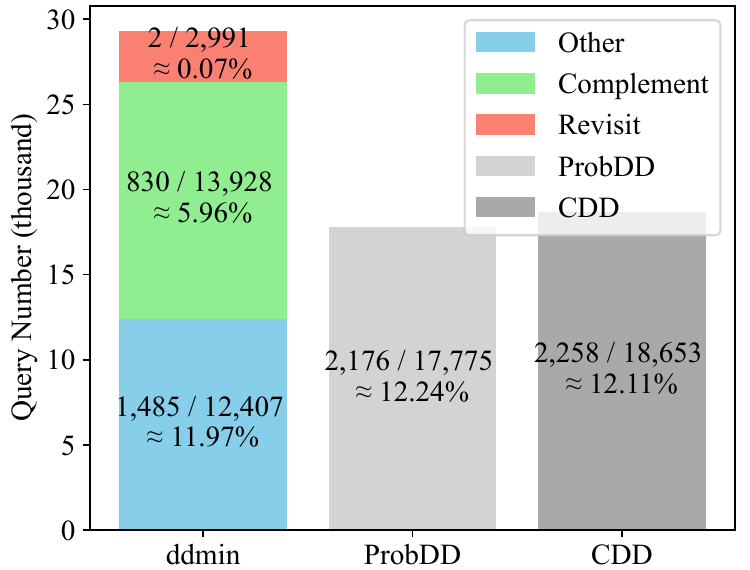}
    	\caption{ On \xmlBenchmarkNameShort }
    	\label{fig:xml_stacked_barchart}
    \end{subfigure}
	\caption{
		Visualization of queries within \ddmin, \probdd and
		\proj. In \ddmin,
		three types of queries are displayed via stacked bars, the height of
		which denotes the query number. Within each bar, the number
		of successful queries, total queries and the
		corresponding success rate are annotated.
%
	}
	\label{fig:stacked_barchart}
\end{figure*}

To further understand why \probdd is more efficient than
\ddmin, we
conduct in-depth statistical analysis on the query number (number of deletion
attempts).
Intuitively, performance bottlenecks lie in those queries with low success
rates, impairing \ddmin's efficiency.
Existing studies \cite{hodovan2016practical, gharachorlu2018avoiding}
also demonstrate the presence of queries with low success rates.
Therefore, to qualitatively and quantitatively identify the exact bottlenecks
impairing \ddmin, we statistically analyze all the queries in \ddmin and
categorize them into three types:
\begin{enumerate}[topsep=0pt, leftmargin=*]
	\item \typeComplement: Queries attempting to remove the complement
	of a
	subset.
	According to \ddmin algorithm, given a
	subset (smaller than half of the list \listinput), it attempts to remove
	either the subset or its complement.
	However, evidence~\cite{hodovan2016practical} shows that keeping a
	small subset and
	removing its complement is not likely to succeed,
	especially on
	structured inputs like programs.

	\item \typeRepeated:
	Queries attempting to remove the previously tried
	subset.
	After removing a subset, \ddmin restarts
	the process
	from the
	first subset, leading to repeated deletion attempts on earlier subsets.
	Although the removal of one subset may allow another
	subset to be removable, such repetitions rarely succeed
	and thus offer limited improvement for the
	reduction~\cite{gharachorlu2018avoiding}.
	\item \typeOther: All other queries.
\end{enumerate}
In addition to categorizing queries in \ddmin into the above types, we also
calculate the success rate of each type, aiming to reveal the
bottlenecks of \ddmin.
\cref{fig:stacked_barchart} illustrates the distribution of queries for all
types within \ddmin, as well as the query number for \probdd across all
three benchmark suites.

On all benchmark suites, the number of successful queries in
\ddmin and \probdd is remarkably similar, especially when contrasted with
the substantial difference in the total number of queries.
Specifically, on \cBenchmarkNameShort, \ddmin achieves
$\cRepeatedNumberSuccessDdmin +
\cComplementNumberSuccessDdmin + \cOtherNumberSuccessDdmin =
\cSuccessNumberDdmin$ successful queries, closely matching the
\cSuccessNumberProbdd successful queries from \probdd. Similarly, on
\debloatBenchmarkNameShort and \xmlBenchmarkNameShort, \ddmin
performs $\debloatRepeatedNumberSuccessDdmin +
\debloatComplementNumberSuccessDdmin +
\debloatOtherNumberSuccessDdmin = \debloatSuccessNumberDdmin$
and $\xmlRepeatedNumberSuccessDdmin +
\xmlComplementNumberSuccessDdmin + \xmlOtherNumberSuccessDdmin
= \xmlSuccessNumberDdmin$ successful queries, respectively, both
closely aligning with the \debloatSuccessNumberProbdd and
\xmlSuccessNumberProbdd successful queries achieved by \probdd.
Besides,
\ddmin always performs
significantly more failed queries, resulting in a larger total
query number and thus a longer execution time, as previously discussed in
\cref{subsec:reproduction}.

On all benchmark suites, a large portion of \ddmin's queries
is categorized as \typeComplement and \typeRepeated; however, they both
have a notably low success rate. For instance, on
\cBenchmarkNameShort, out of a total of \cTotalNumberDdmin queries,
\typeComplement and \typeRepeated account for
\cComplementNumberDdmin (\cComplementPercentageDdmin) and
\cRepeatedNumberDdmin (\cRepeatedPercentageDdmin),
respectively. Within such queries in \typeComplement and \typeRepeated,
merely
\cComplementNumberSuccessDdmin (\cComplementSuccessRateDdmin)
and \cRepeatedNumberSuccessDdmin (\cRepeatedSuccessRateDdmin)
queries succeed, \ie, only a tiny portion of attempts
successfully reduce
elements. These success rates are far less than those of queries within
\typeOther (\cOtherSuccessRateDdmin), as well as those of \probdd
(\cSuccessRateProbdd). On the other benchmark suites, a similar
phenomenon is observed.

Queries within \typeComplement and \typeRepeated categories constitute
a large portion yet prove to be largely inefficient,
wasting a significant amount of time and resources. On the contrary,
those in \typeOther achieve a much higher success rate, on par with that
of \probdd, and
are responsible for most of the successful deletions.
Therefore, we
believe that these two categories, where queries are inefficient, are
the
main bottlenecks
behind \ddmin's low
efficiency. However, these bottlenecks are absent in \probdd,
as it does not consider complements of subsets and
previously tried subsets for deletion.

\subsection{1-Minimality  of \probdd?}
\label{subsec:finding_limitation}
\finding{Improving efficiency by avoiding ineffective
	attempts
	presents a trade-off by not ensuring 1-minimality, while such limitation
	can be mitigated by iteratively running the reduction algorithm until a
	fixpoint is
	reached.}



Although \probdd avoids \typeRepeated queries to enhance efficiency,
some reduction potentials may be missed, as the deletion of a certain
subset may enable a previously tried subset to become removable.
Therefore, a limitation of \probdd lies in that it increases efficiency by
sacrificing 1-minimality. To substantiate this limitation, we examine how
frequently \probdd generates a list that is not 1-minimal, \ie, can be further
reduced by removing a single element. For instance, statistical analysis on
\cBenchmarkNameShort reveals that among
\cBenchmarkProbddCalledNum invocations of \probdd,
\cBenchmarkProbddNotoneminimalNum of them fail to
generate a 1-minimal result, accounting for
\cBenchmarkProbddNotoneminimalRate.
For these failed invocations, an average of
\cBenchmarkProbddNotoneminimalAvgPotential elements (tree nodes) can
be further
removed via single-element deletion.

However, such limitation is not apparent across all
 benchmark suites, as the results from \probdd
 are not
 consistently larger than those from \ddmin. Our further investigation
 reveals that these
 benchmarks are reduced on
 wrapper frameworks
 \picireny and \chisel. Both
 frameworks employ iterative loops to achieve a fixpoint, effectively
 reducing some elements missed in the first iteration.


\section{Implications: a counter-based model}
\label{sec:implication}
Building on the aforementioned demystification of \probdd, we
discover
that probability can be optimized away, and subset size can be pre-computed.
Hence, we propose \projfull (\proj), to reduce the
complexity of both
the theory and
implementation of \probdd, and validate the correctness of our prior
theoretical
proofs.

\begin{algorithm}[h]
    \small
	\DontPrintSemicolon
	\SetKwInput{KwData}{Input}
	\SetKwInput{KwResult}{Output}
	\caption{\protect\AlgCDD{$L, \psi$}}
	\label{alg:counterdd_main}

	\SetKwRepeat{Do}{do}{while}
	\KwData{$L$: a list of element to be reduced.}
	\KwData{$\psi: \searchspace \rightarrow \boolspace$:
		the property to be preserved by $L$.
	}
	\KwData{\pinit:
		the initial probability given by the user.
	}
	\KwResult{the minimized list that still exhibits the property $ \psi $ .}

	$\RoundNumber \gets 0$ \tcp{The round number, initially 0.}
	\label{alg:counterdd_main:begin}
	\Do{$\mysize{}{} > 1$}{
		\tcp{Compute subset size by round number}
		$\mysize{}{} \gets$ \AlgComputeSize(\RoundNumber, \pinit) \\
		\label{alg:counterdd_main:compute_size}
		\tcc{Partition L into subsets with \mysize{}{} elements. If
			it does not divide evenly, leave a smaller remainder as the final
			subset.}
		$\texttt{subsets} \gets$ \AlgPartition(L, \mysize{}{}) \\
		\label{alg:counterdd_main:partition}
		\ForEach{$\texttt{subset} \in \texttt{subsets}$}{
			\label{alg:counterdd_main:inner_loop_begin}
			$\texttt{temp} \gets L \setminus \texttt{subset}$ \\
			\tcp{Remove $subset$ if it is removable}
			\If{$\psi(\texttt{temp})$ is true}{
				$L \gets \texttt{temp}$\;
			}
			\label{alg:counterdd_main:inner_loop_end}
		}

		\tcp{Update the \RoundNumber and move to
		next round.}
		$\RoundNumber \gets \RoundNumber + 1$
		\\
	}
	\Return $L$ \; \label{alg:counterdd_main:end}

	\Fn{\AlgComputeSize{$\RoundNumber, \myprobability{0}{}$}}{
		\KwData{\RoundNumber: the current round number.}
		\KwData{\myprobability{0}{}: the initial probability given by the user.}
		\KwResult{The size of the subset to be used in the current round.}
		\tcp{Calculate the estimated probability of round \RoundNumber}
		\label{alg:computesize:begin}
		$\myprobability{r}{} \gets \myprobability{0}{} \times
		\pincreaseratevalue^{\RoundNumber}$ \\
		\tcp{Calculate corresponding subset size of round \RoundNumber}
		$\mysize{r}{} = \operatorname*{arg\,max}_{\mysize{}{} \in
		\mathbb{N}^{+}} \mysize{}{} \times (1 -
		\myprobability{r}{})^{\mysize{}{}}
		$ \\
		\Return \mysize{r}{} \;
		\label{alg:computesize:end}
	}

\end{algorithm}

\myparagraph{Subset size pre-calculation}
Based on  \cref{lem:s:relation} in  \cref{subsec:finding_size},
the size
for each round can be pre-calculated. Therefore, as shown at
 \cref{alg:computesize:begin} --
 \cref{alg:computesize:end} in  \cref{alg:counterdd_main}, we utilize the
current round \RoundNumber and the initial probability
\pinit to
determine the subset size \mysize{}{}.
The size of the selected subset decreases as the round counter increases.
This is intuitively reasonable since, after a sufficient number of attempts on a large size have been made,
it becomes more advantageous to gradually reduce the subset size for future trials.
Furthermore, this trend aligns well with that of \probdd, in which
probabilities of elements gradually increase, resulting in a smaller subset
size.

\myparagraph{Main workflow}
The simplified \probdd is illustrated in  \cref{alg:counterdd_main}, from
 \cref{alg:counterdd_main:begin} to  \cref{alg:counterdd_main:end}.
Before each round, the \proj pre-calculates the subset
size on
 \cref{alg:counterdd_main:compute_size}
and
then partitions \listinput using this size on
 \cref{alg:counterdd_main:partition}.
Then, similar to \ddmin,
it attempts
to remove each subset on
 \cref{alg:counterdd_main:inner_loop_begin} --
 \cref{alg:counterdd_main:inner_loop_end}.
The subset size continuously decreases until it reaches 1, meaning that
each
element will be individually removed once.

%

\myparagraph{Revisiting the running example}
Returning to \cref{table:running_example_results_probdd}, under the same
conditions, \proj achieves the same results as \probdd but without the need
for probability calculations. This is because both the probability and subset
size \mysize{}{} can be directly determined from the round number
\RoundNumber.

\myparagraph{Evaluation}
As shown in \cref{table:rq1_firsthalf} and \cref{table:rq1_secondhalf}, \proj
outperforms \ddmin \wrt efficiency, with
\allBenchmarkNoTimeoutTimeDecRateProjVsDdmin less time and
\allBenchmarkNoTimeoutQueryDecRateProjVsDdmin fewer queries.
Meanwhile, \proj performs comparably to \probdd \wrt final size, execution
time and
query number,
with a p-value of \allBenchmarkSizeProbddVsProjPvalue,
\allBenchmarkTimeProbddVsProjPvalue and
\allBenchmarkQueryProbddVsProjPvalue, respectively, indicating
insignificance between these
two
algorithms. \proj is expected to perform on par with \probdd since it is
designed to provide further insight and simplify the intricate design of
\probdd, rather than to surpass its capabilities.
Furthermore, its comparable performance to \probdd further validates the
non-necessity of randomness and
our assumption in \cref{lem:s:uniform-prob}.

\myparagraph{Bottleneck and 1-minimality}
Revisiting the bottlenecks presented in \cref{fig:stacked_barchart}, \proj
possesses a query number and success rate close to those of \probdd,
indicating that \proj also overcomes the bottlenecks of \ddmin. Additionally,
similar to \probdd, 1-minimality is absent in \proj, although iterations help
mitigate this issue.

\finding{\proj always
	achieves
	comparable performance to \probdd,
   	which further supports our previous findings,
     including the theoretical simplifications regarding size and probability,
     analysis of randomness, bottlenecks, and 1-minimality.
}

\section{Limitations and threats to validity}
\label{sec:limitation}
In this section, we discuss the limitations of \proj,
and potential factors that may impair the validity of our experimental
results.

\subsection{Limitations}
As discussed in \cref{subsec:finding_limitation},
compared to \ddmin, neither \probdd nor \proj guarantees 1-minimality.
This limitation arises because after successfully removing a subset, \ddmin
restarts the process from the first subset, whereas \probdd and \proj
continue
from the next subset, skipping all the previously tried subsets.
Therefore, although \probdd and \proj complete the reduction process
more quickly,
they may miss certain reduction opportunities and produce larger results
than \ddmin.

However, reduction and debloating tools generally invoke
these reduction algorithms in iterative loops until a fix-point is reached,
gradually refining the results and mitigating limitations, as mentioned in
\cref{subsec:finding_limitation}.
\cref{table:rq1_firsthalf} and \cref{table:rq1_secondhalf} further support this point by showing that with multiple
iterations,
\probdd and \proj achieve significantly higher efficiency compared to
\ddmin,
while still producing results that are comparable to \ddmin w.r.t.
effectiveness.

\subsection{Threats to validity}
For internal validity, the main threat comes from the potential impact from
the assumption, as discussed in \cref{subsection:assumption}. Specifically,
without assuming that the number of elements in \listinput is always
divisible by the current subset size, we could not further refine the
mathematical model of \probdd to achieve a simpler representation.
However, such assumption might impact the actual performance,
potentially negating the benefits of our simplification. To this end, we
conduct extensive empirical experiments, demonstrating that \proj, the
simplified algorithm derived from this assumption, exhibits no significant
difference from \probdd.

For external validity, the threat lies in the generalizability of our findings
across
application scenarios. To mitigate this threat, we perform experiments
on \benchmarkNumTotal benchmarks, including C programs triggering
real-world compiler bugs, \xml inputs crashing \xml processing tools, and
benchmarks from software debloating tasks. These benchmarks have
covered various use scenarios of minimization algorithms.

\section{Related Work}
\label{sec:related_work}
In this section, we discuss related work of test input minimization around
three aspects: effectiveness, efficiency, and the utilization of domain
knowledge.

\myparagraph{Effectiveness}
Test input minimization is an NP-complete problem, in which achieving the
global minimum is usually infeasible. Therefore, existing approaches to
improving effectiveness mainly aim to escape local minima by performing
more exhaustive searches. Since enumerating all
possible subsets is infeasible,
\vulcan~\cite{xu2023pushing} and \creduce~\cite{regehr2012test}
enumerate all combinations of elements
within a small sliding window, and exhaustively attempt to delete each
combination,
resulting in smaller final program sizes.
In contrast, \probdd and \proj do not exhibit clear
actions targeted at breaking through local optima, suggesting they cannot
achieve better effectiveness than \ddmin, as aligned with our evaluation in
\cref{sec:experimental_results}.

\myparagraph{Efficiency}
If parallelism is not considered, the core of boosting efficiency is the
enhanced capability to avoid relatively
inefficient queries.
For example, Hodovan and Kiss~\cite{hodovan2016practical} proposed
disregarding attempts to remove the
complement of subsets, the success rate of which is unacceptably low in
some
scenarios.
Besides,
Gharachorlu and Sumner~\cite{gharachorlu2018avoiding} proposed One
Pass Delta Debugging (OPDD), which continues with
the subset next to the deleted one, rather than starting over from the first
subset. This optimization also avoids some redundant queries in
\ddmin, reducing runtime by 65\%.
As revealed by our analysis, these
two above-mentioned
optimizations are implicitly incorporated within \probdd and \proj, and
thereby contributing to their higher efficiency than \ddmin.

\myparagraph{Utilization of domain knowledge}
There is an inherent trade-off between effectiveness
and efficiency in test input minimization. For the same algorithm, achieving
a better result, \ie, a smaller local optimum, requires more queries to be
spent on trial and error.
However, employing domain knowledge~\cite{regehr2012test,
	kalhauge2019binary, kalhauge2021logical, zhang2024lpr} can still
	improve the overall
performance. For instance, \jreduce is both
more
effective and efficient than \hdd
in reducing \java programs, as it escapes more
local optima by program transformations while simultaneously avoiding
more
inefficient
queries via semantic constraints, leveraging the semantics of \java.
Our analysis on \probdd indicates that the probabilities primarily
function as
counters and do not
utilize or effectively learn the domain knowledge
of an input. Besides, the evaluation on \proj, a simplified algorithm without
utilizing
probability, demonstrates that prioritizing
elements
via such
probabilities does not yield significant benefits,
thus validating our analysis.

\section{Conclusion}
\label{sec:conclusion}

 This paper conducts the first in-depth analysis of \probdd, which is the
 state-of-the-art variant of \ddmin, to further comprehend and demystify
 its superior performance. With theoretical analysis of the probabilistic
 model in \probdd, we reveal
 that probabilities essentially serve as monotonically increasing counters,
 and propose \proj for simplification. Evaluations on \benchmarkNumTotal
 benchmarks from test input minimization
 and software debloating
 confirm that \proj performs on par with \probdd, substantiating our
 theoretical analysis. Furthermore, our examination on query
 success rate
 and randomness uncovers that \probdd's superiority stems from skipping
 inefficient queries.
 Finally, we discuss trade-offs in \ddmin and \probdd, providing insights
 for future research and applications of test input minimization algorithms.

\section*{Acknowledgments}

We thank all the anonymous reviewers in ICSE'25 for their insightful
feedback and comments.
This research is partially supported by
the Natural Sciences and Engineering Research Council of Canada
(NSERC) through the
Discovery Grant, a project under WHJIL,
and CFI-JELF Project \#40736.

\balance
\bibliographystyle{IEEEtran}
\bibliography{reference}

\begin{thebibliography}{10}
\providecommand{\url}[1]{#1}
\csname url@samestyle\endcsname
\providecommand{\newblock}{\relax}
\providecommand{\bibinfo}[2]{#2}
\providecommand{\BIBentrySTDinterwordspacing}{\spaceskip=0pt\relax}
\providecommand{\BIBentryALTinterwordstretchfactor}{4}
\providecommand{\BIBentryALTinterwordspacing}{\spaceskip=\fontdimen2\font plus
\BIBentryALTinterwordstretchfactor\fontdimen3\font minus
  \fontdimen4\font\relax}
\providecommand{\BIBforeignlanguage}[2]{{%
\expandafter\ifx\csname l@#1\endcsname\relax
\typeout{** WARNING: IEEEtran.bst: No hyphenation pattern has been}%
\typeout{** loaded for the language `#1'. Using the pattern for}%
\typeout{** the default language instead.}%
\else
\language=\csname l@#1\endcsname
\fi
#2}}
\providecommand{\BIBdecl}{\relax}
\BIBdecl

\bibitem{zeller2002simplifying}
A.~Zeller and R.~Hildebrandt, ``Simplifying and isolating failure-inducing
  input,'' \emph{IEEE Transactions on Software Engineering}, vol.~28, no.~2,
  pp. 183--200, 2002.

\bibitem{gccReductionGuide}
\BIBentryALTinterwordspacing
GCC. (2020) A guide to testcase reduction. Accessed: 2023-04-30. [Online].
  Available: \url{https://gcc.gnu.org/wiki/A_guide_to_testcase_reduction}
\BIBentrySTDinterwordspacing

\bibitem{llvmReductionGuide}
\BIBentryALTinterwordspacing
LLVM. (2022) How to submit an llvm bug report. Accessed: 2023-04-30. [Online].
  Available: \url{https://llvm.org/docs/HowToSubmitABug.html}
\BIBentrySTDinterwordspacing

\bibitem{webkit}
\BIBentryALTinterwordspacing
WebKit. (2001) Webkit: Test case reduction. Accessed: 2023-04-30. [Online].
  Available: \url{https://webkit.org/test-case-reduction/}
\BIBentrySTDinterwordspacing

\bibitem{asfbugzilla}
\BIBentryALTinterwordspacing
{ASF Bugzilla}. (2001) {ASF} bugzilla: Bug writing guidelines. Accessed:
  2023-04-30. [Online]. Available:
  \url{https://bz.apache.org/bugzilla/page.cgi?id=bug-writing.html}
\BIBentrySTDinterwordspacing

\bibitem{mozillabugzilla}
\BIBentryALTinterwordspacing
Bugzilla. (2001) Bugzilla: Reporting a new bug. Accessed: 2023-04-30. [Online].
  Available:
  \url{https://bugzilla.readthedocs.io/en/5.2/using/filing.html#reporting-a-new-bug}
\BIBentrySTDinterwordspacing

\bibitem{sigplan}
\BIBentryALTinterwordspacing
A.~Donaldson and D.~MacIver. (2021, May) {Test Case Reduction: Beyond Bugs}.
  [Online]. Available:
  \url{https://blog.sigplan.org/2021/05/25/test-case-reduction-beyond-bugs}
\BIBentrySTDinterwordspacing

\bibitem{ccmd}
T.~L. Wang, Y.~Tian, Y.~Dong, Z.~Xu, and C.~Sun, ``Compilation consistency
  modulo debug information,'' in \emph{{ASPLOS} '23: 28th {ACM} International
  Conference on Architectural Support for Programming Languages and Operating
  Systems, Lausanne, Vancouver, 25 March 2023 - 29 March 2023}.\hskip 1em plus
  0.5em minus 0.4em\relax {ACM}, 2023.

\bibitem{kitten}
\BIBentryALTinterwordspacing
Y.~Tian, Z.~Xu, Y.~Dong, C.~Sun, and S.~Cheung, ``Revisiting the evaluation of
  deep learning-based compiler testing,'' in \emph{Proceedings of the
  Thirty-Second International Joint Conference on Artificial Intelligence,
  {IJCAI} 2023, 19th-25th August 2023, Macao, SAR, China}.\hskip 1em plus 0.5em
  minus 0.4em\relax ijcai.org, 2023, pp. 4873--4882. [Online]. Available:
  \url{https://doi.org/10.24963/ijcai.2023/542}
\BIBentrySTDinterwordspacing

\bibitem{sun2016finding}
C.~Sun, V.~Le, and Z.~Su, ``Finding compiler bugs via live code mutation,'' in
  \emph{Proceedings of the 2016 ACM SIGPLAN International Conference on
  Object-Oriented Programming, Systems, Languages, and Applications}, 2016, pp.
  849--863.

\bibitem{donaldson2021test}
A.~F. Donaldson, P.~Thomson, V.~Teliman, S.~Milizia, A.~P. Maselco, and
  A.~Karpi{\'n}ski, ``Test-case reduction and deduplication almost for free
  with transformation-based compiler testing,'' in \emph{Proceedings of the
  42nd ACM SIGPLAN International Conference on Programming Language Design and
  Implementation}, 2021, pp. 1017--1032.

\bibitem{sun2018perses}
C.~Sun, Y.~Li, Q.~Zhang, T.~Gu, and Z.~Su, ``Perses: Syntax-guided program
  reduction,'' in \emph{Proceedings of the 40th International Conference on
  Software Engineering}, 2018, pp. 361--371.

\bibitem{misherghi2006hdd}
G.~Misherghi and Z.~Su, ``Hdd: hierarchical delta debugging,'' in
  \emph{Proceedings of the 28th International Conference on Software
  Engineering}, 2006, pp. 142--151.

\bibitem{regehr2012test}
J.~Regehr, Y.~Chen, P.~Cuoq, E.~Eide, C.~Ellison, and X.~Yang, ``Test-case
  reduction for c compiler bugs,'' in \emph{Proceedings of the 33rd ACM SIGPLAN
  Conference on Programming Language Design and Implementation}, 2012, pp.
  335--346.

\bibitem{heo2018effective}
K.~Heo, W.~Lee, P.~Pashakhanloo, and M.~Naik, ``Effective program debloating
  via reinforcement learning,'' in \emph{Proceedings of the 2018 ACM SIGSAC
  Conference on Computer and Communications Security}, 2018, pp. 380--394.

\bibitem{wang2021probabilistic}
G.~Wang, R.~Shen, J.~Chen, Y.~Xiong, and L.~Zhang, ``Probabilistic delta
  debugging,'' in \emph{Proceedings of the 29th ACM Joint Meeting on European
  Software Engineering Conference and Symposium on the Foundations of Software
  Engineering}, 2021, pp. 881--892.

\bibitem{gharachorlu2018avoiding}
G.~Gharachorlu and N.~Sumner, ``Avoiding the familiar to speed up test case
  reduction,'' in \emph{2018 IEEE International Conference on Software Quality,
  Reliability and Security (QRS)}.\hskip 1em plus 0.5em minus 0.4em\relax IEEE,
  2018, pp. 426--437.

\bibitem{hodovan2016practical}
R.~Hodov{\'a}n and {\'A}.~Kiss, ``Practical improvements to the minimizing
  delta debugging algorithm.'' in \emph{ICSOFT-EA}, 2016, pp. 241--248.

\bibitem{zhou2024wdd}
X.~Zhou, Z.~Xu, M.~Zhang, Y.~Tian, and C.~Sun, ``Wdd: Weighted delta
  debugging,'' in \emph{Proceedings of the IEEE/ACM 47th International
  Conference on Software Engineering}, 2025.

\bibitem{probdd}
\BIBentryALTinterwordspacing
G.~Wang. (2021) Probdd. Accessed: 2023-04-30. [Online]. Available:
  \url{https://github.com/Amocy-Wang/ProbDD}
\BIBentrySTDinterwordspacing

\bibitem{zhang2024cddartifact}
\BIBentryALTinterwordspacing
M.~Zhang, Z.~Xu, Y.~Tian, X.~Cheng, and C.~Sun, ``Artifact for "toward a better
  understanding of probabilistic delta debugging",'' 2024. [Online]. Available:
  \url{https://zenodo.org/records/14425530}
\BIBentrySTDinterwordspacing

\bibitem{christi2017resource}
A.~Christi, A.~Groce, and R.~Gopinath, ``Resource adaptation via test-based
  software minimization,'' in \emph{2017 IEEE 11th International Conference on
  Self-Adaptive and Self-Organizing Systems (SASO)}.\hskip 1em plus 0.5em minus
  0.4em\relax IEEE, 2017, pp. 61--70.

\bibitem{groce2014cause}
A.~Groce, M.~A. Alipour, C.~Zhang, Y.~Chen, and J.~Regehr, ``Cause reduction
  for quick testing,'' in \emph{2014 IEEE Seventh International Conference on
  Software Testing, Verification and Validation}.\hskip 1em plus 0.5em minus
  0.4em\relax IEEE, 2014, pp. 243--252.

\bibitem{groce2016cause}
------, ``Cause reduction: delta debugging, even without bugs,'' \emph{Software
  Testing, Verification and Reliability}, vol.~26, no.~1, pp. 40--68, 2016.

\bibitem{picire}
\BIBentryALTinterwordspacing
A.~Kiss, R.~Hodován, and D.~Vince. (2016) Picire. Accessed: 2023-04-30.
  [Online]. Available: \url{https://github.com/renatahodovan/picire}
\BIBentrySTDinterwordspacing

\bibitem{pelikan1999boa}
M.~Pelikan, D.~E. Goldberg, E.~Cant{\'u}-Paz \emph{et~al.}, ``Boa: The bayesian
  optimization algorithm,'' in \emph{Proceedings of the genetic and
  evolutionary computation conference GECCO-99}, vol.~1.\hskip 1em plus 0.5em
  minus 0.4em\relax Citeseer, 1999, pp. 525--532.

\bibitem{vulcan}
\BIBentryALTinterwordspacing
Z.~Xu, Y.~Tian, M.~Zhang, G.~Zhao, Y.~Jiang, and C.~Sun, ``Pushing the limit of
  1-minimality of language-agnostic program reduction,'' \emph{Proc. ACM
  Program. Lang.}, vol.~7, no. OOPSLA1, apr 2023. [Online]. Available:
  \url{https://doi.org/10.1145/3586049}
\BIBentrySTDinterwordspacing

\bibitem{trec}
\BIBentryALTinterwordspacing
Z.~Xu, Y.~Tian, M.~Zhang, J.~Zhang, P.~Liu, Y.~Jiang, and C.~Sun, ``T-rec:
  Fine-grained language-agnostic program reduction guided by lexical syntax,''
  \emph{ACM Trans. Softw. Eng. Methodol.}, Aug. 2024, just Accepted. [Online].
  Available: \url{https://doi.org/10.1145/3690631}
\BIBentrySTDinterwordspacing

\bibitem{qian2019razor}
C.~Qian, H.~Hu, M.~Alharthi, S.~P.~H. Chung, T.~Kim, and W.~Lee, ``Razor: A
  framework for post-deployment software debloating.'' in \emph{USENIX Security
  Symposium}, 2019, pp. 1733--1750.

\bibitem{alhanahnah2022lightweight}
M.~Alhanahnah, R.~Jain, V.~Rastogi, S.~Jha, and T.~Reps, ``Lightweight,
  multi-stage, compiler-assisted application specialization,'' in \emph{2022
  IEEE 7th European Symposium on Security and Privacy (EuroS\&P)}.\hskip 1em
  plus 0.5em minus 0.4em\relax IEEE, 2022, pp. 251--269.

\bibitem{li2024finding}
S.~Li and M.~Rigger, ``Finding xpath bugs in xml document processors via
  differential testing,'' \emph{arXiv preprint arXiv:2401.05112}, 2024.

\bibitem{woolson2005wilcoxon}
R.~F. Woolson, ``Wilcoxon signed-rank test,'' \emph{Encyclopedia of
  Biostatistics}, vol.~8, 2005.

\bibitem{picireny}
\BIBentryALTinterwordspacing
A.~Kiss, R.~Hodován, and D.~Vince. (2016) Picireny. Accessed: 2023-04-30.
  [Online]. Available: \url{https://github.com/renatahodovan/picireny}
\BIBentrySTDinterwordspacing

\bibitem{xu2023pushing}
Z.~Xu, Y.~Tian, M.~Zhang, G.~Zhao, Y.~Jiang, and C.~Sun, ``Pushing the limit of
  1-minimality of language-agnostic program reduction,'' \emph{Proceedings of
  the ACM on Programming Languages}, vol.~7, no. OOPSLA1, pp. 636--664, 2023.

\bibitem{kalhauge2019binary}
C.~G. Kalhauge and J.~Palsberg, ``Binary reduction of dependency graphs,'' in
  \emph{Proceedings of the 2019 27th ACM Joint Meeting on European Software
  Engineering Conference and Symposium on the Foundations of Software
  Engineering}, 2019, pp. 556--566.

\bibitem{kalhauge2021logical}
------, ``Logical bytecode reduction,'' in \emph{Proceedings of the 42nd ACM
  SIGPLAN International Conference on Programming Language Design and
  Implementation}, 2021, pp. 1003--1016.

\bibitem{zhang2024lpr}
M.~Zhang, Y.~Tian, Z.~Xu, Y.~Dong, S.~H. Tan, and C.~Sun, ``{LPR}: Large
  language models-aided program reduction,'' in \emph{Proceedings of the 33rd
  ACM SIGSOFT International Symposium on Software Testing and Analysis}.\hskip
  1em plus 0.5em minus 0.4em\relax New York, NY, USA: ACM, 2024, p.~13.

\end{thebibliography}
\end{document}